\documentclass[final]{svjour3}
\usepackage[draft]{fixme}

\usepackage{graphicx}

\usepackage{amssymb}
\usepackage{algorithmic}
\usepackage{algorithm}
\usepackage{bbding}
\usepackage{enumitem}
\usepackage{subfigure}
\usepackage{url}

\usepackage[utf8]{inputenc}

\usepackage{amsmath}
\usepackage{amsfonts}
\usepackage{amssymb}
\usepackage{mathrsfs}
\usepackage{graphicx}
\usepackage{stackrel}
\usepackage{nicefrac}

\input xy
\xyoption{all}

\newcommand{\argmin}{\mathrm{argmin}}
\newcommand{\argmax}{\mathrm{argmax}}

\newcommand{\RR}{\mathbb{R}}
\newcommand{\CC}{\mathbb C}

\newcommand{\Span}{\mathrm{span}\,}

\newcommand{\SSS}{\mathbb{S}}

\renewcommand{\phi}{\varphi}

\newcommand{\pfc}[2]{\tfrac{\partial}{\partial #2}#1}
\newcommand{\pfcc}[3]{\tfrac{\partial^2}{\partial #2\partial #3}#1}

\newcommand{\df}[2]{\tfrac{d}{d#2} #1}
\newcommand{\sdf}[2]{\scriptscriptstyle\frac{d}{d#2} #1}
\newcommand{\Df}[2]{\tfrac{D}{d#2} #1}
\newcommand{\tm}{d}

\newcommand{\grad}{\mathrm{grad}}

\newcommand{\ip}[1]{\left<#1\right>}

\newcommand{\Exp}{\mathrm{Exp}}
\newcommand{\Log}{\mathrm{Log}}

\newcommand{\lemref}[1]{\text{Lemma~\ref{lem:#1}}}

\newcommand{\tmp}[1]{\tm^{#1}}
\newcommand{\Hess}{H}
\newcommand{\param}[1]{\mathbf{#1}}

\begin{document}

\title{Optimization over Geodesics for Exact Principal Geodesic Analysis
}
\titlerunning{Optimization over Geodesics for Exact Principal Geodesic Analysis}

\author{S. Sommer
\and
F. Lauze
\and
M. Nielsen
}

\institute{S. Sommer (\Envelope) \at
              Dept. of Computer Science, Univ. of Copenhagen, Copenhagen, Denmark\\
              \email{sommer@diku.dk}, 
              Tel.: +4535321400
           \and
F. Lauze \at
              Dept. of Computer Science, Univ. of Copenhagen, Copenhagen, Denmark\\
              \email{francois@diku.dk}
           \and
M. Nielsen \at
              Dept. of Computer Science, Univ. of Copenhagen, Copenhagen, Denmark\\
              Biomediq, Copenhagen, Denmark, 
              \email{madsn@diku.dk}
}

\maketitle

\begin{abstract}
In fields ranging from computer vision to signal processing and statistics, increasing computational power allows a move from classical linear models to models that incorporate non-linear phenomena. This shift has created interest in computational aspects of differential geometry, and solving optimization problems that incorporate non-linear geometry constitutes an important computational task. In this paper, we develop methods for numerically solving optimization problems over spaces of geodesics using numerical integration of Jacobi fields and second order derivatives of geodesic families. As an important application of this optimization strategy, we compute exact Principal Geodesic Analysis (PGA), a non-linear version of the PCA dimensionality reduction procedure. By applying the exact PGA algorithm to synthetic data, we exemplify the differences between the linearized and exact algorithms caused by the non-linear geometry. In addition, we use the numerically integrated Jacobi fields to determine sectional curvatures and provide upper bounds for injectivity radii. 

\keywords{
    geometric optimization,
    principal geodesic analysis, 
    manifold statistics,
    differential geometry, 
    Riemannian metrics
    }
    \subclass{65K10 \and 57R99}
\end{abstract}


\section{Introduction}
\label{introduction}
Manifolds, sets locally modeled by Euclidean spaces, have a long and intriguing
history in mathematics, and topological, differential geometric, and Riemannian
geometric properties of manifolds have been studied extensively. The introduction of
high performance computing in applied fields has widened the use of differential
geometry,
and Riemannian manifolds, in particular, are now used for modeling a range
of problems possessing non-linear structure. Applications include shape modeling 
(complex projective shape spaces
\cite{kendall_shape_1984} and medial representations of surfaces
\cite{blum_transformation_1967,joshi_multiscale_2002}),
imaging (tensor manifolds in diffusion tensor imaging 
\cite{fletcher_principal_2004,fletcher_riemannian_2007,pennec_riemannian_2006} and
image segmentation and registration \cite{caselles_geodesic_1995,pennec_feature-based_1998}),
and several other fields (forestry \cite{huckemann_intrinsic_2010}, human motion
modeling \cite{sminchisescu_generative_2004,urtasun_priors_2005},
information geometry and signal processing \cite{yang_means_2011}).

To fully utilize the power of manifolds in applied modeling, it is essential to develop fast and
robust algorithms for performing computations on manifolds, and, in particular,
availability of methods for solving optimization problems is paramount. In this paper, 
we develop methods for numerically solving optimization problems over spaces of
geodesics using numerical integration of Jacobi fields and second order derivatives of
geodesic families.
In addition, the approach allows numerical approximation of sectional curvatures
and bounds on injectivity
radii \cite{huckemann_intrinsic_2010}. The methods apply to
manifolds represented both parametrically and implicitly without 
preconditions such as knowledge of explicit formulas for geodesics. This fact
makes the approach applicable to a range of applications, and it allows
exploration of the effect of curvature on non-linear statistical
methods.

To exemplify this, we consider the problem of capturing the variation of a set
of manifold valued data with the 
Principal Geodesic Analysis (PGA, \cite{fletcher_principal_2004-1}) non-linear
generalization of Principal Component Analysis (PCA). Until now, there has been no
method for numerically computing PGA for general manifolds without linearizing
the problem. Because PGA can be formulated as an optimization
problem over geodesics, the tools developed here apply to computing
it without discarding the non-linear structure. As a result, the paper 
provides an algorithm for
computing exact PGA for a wide range of manifolds.

\subsection{Related Work}
\label{sec:related}
A vast body of mathematical literature describes manifolds and Riemannian
structures; 
\cite{do_carmo_riemannian_1992,lee_riemannian_1997} provide
excellent introductions to the field. From an applied point of view, the papers
\cite{dedieu_symplectic_2005,herbert_bishop_keller_numerical_1968,noakes_global_1998,klassen_geodesics_2006,schmidt_shape_2006,sommer_bicycle_2009}
address first-order problems such as computing geodesics and solving the
exponential map inverse problem, the logarithm map.
Certain second-order problems including computing Jacobi fields on
diffeomorphism groups \cite{younes_evolutions_2009,ferreira_newton_2008}
have been considered but mainly on limited classes of manifolds. 
Different aspects of numerical computation on implicitly defined manifolds are covered in
\cite{zhang_curvature_2007,rheinboldt_manpak:_1996,rabier_computational_1990},
and generalizing linear statistics to manifolds has been the focus of the papers
\cite{karcher_riemannian_1977,pennec_intrinsic_2006,fletcher_robust_2008,fletcher_principal_2004-1,huckemann_intrinsic_2010}.

Optimization problems can be posed on a manifold in the sense that the domain of
the cost function is restricted to the manifold.
Such problems are extensively covered in the literature 
(e.g. \cite{luenberger_gradient_1972,yang_globally_2007}). In contrast, this
paper concerns
optimization problems over geodesics with the complexity residing in 
the cost functions and not the optimization
domains.

The manifold generalization of linear PCA, PGA, was first introduced in
\cite{fletcher_statistics_2003} but it was formulated in the form most widely used in
\cite{fletcher_principal_2004-1}. It has subsequently been used for several
applications. To mention a few, the authors in 
\cite{fletcher_principal_2004-1,fletcher_principal_2004} study
variations of medial atoms, \cite{wu_weighted_2008} uses a variation
of PGA for facial classification, \cite{said_exact_2007} presents examples on
motion capture data, and \cite{sommer_bicycle_2009} applies PGA to vertebrae
outlines. 
The algorithm presented in \cite{fletcher_principal_2004-1} for computing PGA with tangent space
linearization is most widely used. 
In contrast, \cite{said_exact_2007} computes PGA as
defined in \cite{fletcher_statistics_2003} without approximations, exact PGA, on
a particular 
manifold, the Lie group $\mathrm{SO}(3)$. The paper
\cite{sommer_manifold_2010} uses the methods presented here to experimentally
assess the effect of tangent space linearization
on high dimensional manifolds modeling real-life
data.

A recent wave of interest in manifold valued statistics 
has lead to the development of Geodesic PCA (GPCA, 
\cite{huckemann_intrinsic_2010})
and Horizontal Component Analysis (HCA, \cite{sommer_horizontal_2013}).
GPCA is in many respects close to PGA but
GPCA optimizes for the placement of the center point
and minimizes projection residuals along geodesics. HCA builds low-dimensional orthogonal
decompositions in the frame bundle of the manifold that project back to
approximating subspaces in the manifold.

\subsection{Content and Outline}
The paper presents two main contributions: (1) how numerical integration of 
Jacobi fields and second order derivatives can be used to solve optimization problems over
geodesics; and (2) how the approach allows numerical computation of exact PGA. In addition,
we use the computed Jacobi fields to numerically approximate 
geometric properties such as sectional curvatures.
After a brief discussion of the geometric background, explicit differential equations
for computing Jacobi fields and second derivatives of geodesic families are
presented in Section~\ref{sec:differentials}.
The actual derivations are performed in the appendices due to
their notational complexity. In Section~\ref{sec:pga}, the
exact PGA algorithm is developed. We end the paper with experiments that
illustrate the effect of curvature on the non-linear statistical method 
and with estimation of sectional curvatures and injectivity radii bounds.

The importance of curvature computations is noted in
\cite{huckemann_intrinsic_2010}, which lists the ability to compute sectional
curvature as a high importance open problem. The paper presents a partial
solution to this problem: we discuss how sectional curvatures can be determined
numerically when either a parametrization or implicit representation is available.

In the experiments, we evaluate how the differences between the exact
and linearized PGA depend on
the curvature of the manifold. This experiment, which to the best of our
knowledge has not been made before, is made possible by the
generality of the optimization approach that makes the algorithm 
applicable to a wide range of manifolds with varying curvature.

\section{Background}
\label{sec:geom-notation}
This section will include brief discussions of relevant aspects 
of differential and Riemannian geometry. We keep the notation close to the
notation used in the book \cite{do_carmo_riemannian_1992}; 
see in addition Appendix~\ref{app:notation}.

\subsection{Manifolds and Their Representations}
\label{sec:rep}
In the sequel, $M$ will denote a Riemannian manifold of finite dimension $\eta$. We will need $M$ 
to be sufficiently smooth, i.e. of class $C^k$ for $k=3$ or $4$ depending on the
application. For concrete computational applications, we will represent 
$M$ either using \emph{parametrizations} or \emph{implicitly}.
A local parametrization is a map $\param{x}\in C^k(U,M)$ from an open subset
$U\subset\RR^\eta$ to $M$. With an implicit representation,
$M$ is represented as a level set of a differentiable map $F:\RR^m\rightarrow\RR^n$, e.g.
$M=F^{-1}(0)$. If the Jacobian matrix 
$DF$ has full rank $n$ everywhere on $M$, $M$ will be an $(m-n)$-dimensional
manifold. The space $\RR^m$ is called the embedding space.
When dealing with implicitly defined manifolds, we let $m$ and $n$ 
denote the dimension of the domain and codomain of $F$, respectively, so that
the dimension $\eta$ of the manifold equals $m-n$.
Examples of applications using implicit representations include
shape and human poses models \cite{sommer_bicycle_2009,hauberg_natural_2012}, and
several shape models use parametric representations
\cite{joshi_multiscale_2002,klassen_analysis_2004}.\footnote{Other 
representations include discrete triangulations used for surfaces and
quotients $\tilde{M}/G$ of a larger manifold $\tilde{M}$ by a group $G$. The
latter is for example the case for Kendall's shape-spaces $\Sigma_d^k$
\cite{kendall_shape_1984}. Kendall's shape-spaces for planar points are actually
complex projective spaces $\CC P^{k-2}$ for which parameterizations are
available, and, for points in $3$-dimensional space and higher, the shape-spaces
are anomalous and not manifolds. The spaces studied in
\cite{huckemann_intrinsic_2010} belong to this class.}

\subsection{Geodesic Systems}
\label{sec:geodesic-systems}
Given a local parametrization $\param{x}:U\rightarrow M$, a curve $\alpha_t$ on
$M$ is a geodesic if the curve $x_t$ in $U$ representing $\alpha_t$, i.e.
$\param{x}^{-1}\circ\alpha_t=x_t$, satisfies the ODE
\begin{equation}
        \ddot{x}_t^k=-\sum_{i,j}^\eta\Gamma_{ij}^k(x_t)\dot{x_t}^i\dot{x_t}^j,\
        k=1,\ldots,\eta \ .
    \label{eq:param-geo}
\end{equation}
Here $\Gamma_{ij}^k$ denotes the Christoffel symbols of the Riemannian metric.
Conversely, geodesics can be found by solving the ODE with a starting point
$x_0=q$ and initial velocity $\dot{x}_0=v$ as initial conditions. 
The exponential map $\Exp_q v$ maps the initial point $q\in M$ and velocity
$v\in T_q M$ to 
$\alpha_1$, the
point on the geodesic at time $t=1$.
When defined, the logarithm map $\Log_q y$ is the inverse of
$\Exp_q$, i.e. $\Exp_q\Log_q y=y$.
For implicitly represented manifolds, the classical ODE describing geodesics is 
not directly usable because neither parametrizations nor Christoffel symbols 
are directly available. Instead, the geodesic with initial point $q$ and
initial velocity $v$ can be found as the $x$-part of the solution 
of the IVP
\begin{equation}
    \begin{split}
        &\dot{p}_t=-\left(\sum_{k=1}^n\mu^k(x_t,p_t)\Hess_{x_t}(F^k)\right)\dot{x}_t\ ,\\
        &\dot{x}_t=\left(I-D_{x_t}F^\dagger D_{x_t}F\right)p_t\ ,\\
        &x_0=q,\ p_0=v \ ,
    \end{split}
    \label{sys:impl-geo}
\end{equation}
see \cite{dedieu_symplectic_2005}. Note that $x_t$ is a curve in the embedding
space $\RR^m$ but since $M$ is a subset of the embedding space and the starting
point $q$ is in $M$, $x_t$ will stay in $M$ for all $t$.
Recall that $F:\RR^m\rightarrow\RR^n$ is the map defining the manifold by
e.g. $M=F^{-1}(0)$ and that $\Hess(F^k)$ denotes the Hessian of the $k$th
component of $F$. $F$ is map between Euclidean spaces and the Hessian is
therefore the ordinary Euclidean Hessian matrix. 
The map $\mu:\RR^m\times\RR^m\rightarrow\RR^n$ is defined by
$(x,p)\mapsto-(D_{x}F^T)^\dagger p$ where
the symbol $DF^\dagger$ denotes the generalized inverse or pseudo-inverse
of the non-square matrix $DF$. Since $DF$ has full-rank $n$, $DF^\dagger$ equals
$DF^T(DFDF^T)^{-1}$. Numerical stability of the geodesic
system is treated in \cite{dedieu_symplectic_2005}.

\subsection{Geodesic Families and Variations of Geodesics}
In the next sections, we will treat optimization problems over geodesics of which 
the PGA problem \eqref{eq:pga-cost} constitute a concrete example;
in addition, problems such as geodesic regression \cite{fletcher_geodesic_2011} and
manifold total least squares belong to this class.
For this purpose, we here recall the close connection between variations of geodesics,
Jacobi fields, and the differential $\tm\Exp$. 
Let $\alpha_{t,s}$ be a family of geodesics parametrized by $s$, i.e. for each
$\tilde{s}$, the curve $t\mapsto \alpha_{t,\tilde{s}}$ is a geodesic. 
By varying the parameter $s$, a vector field $\df{\alpha_{t,0}}{s}$ is obtained.\footnote{
Recall that $\df{\alpha_{t,0}}{s}$ is a shorthand for
$\df{\alpha_{t,s}}{s}|_{s=0}$, see Appendix~\ref{app:notation}.}
These \emph{Jacobi fields} are uniquely determined by 
the initial conditions $J_0$ and $\Df{J_0}{t}$, the variation of the initial
points $x_{0,s}$ and the covariant derivative of the field at $t=0$, respectively. Define 
$q_s=x_{0,s}$, $v_s=\dot{x}_{0,s}$, and
$w=\df{v_0}{s}$. If $\df{q_0}{s}=J_0$ and $w=\Df{J_0}{t}$ then 
$\df{\Exp_{q_s}(tv_s)}{s}|_{s=0}$ equals $J_t$ \cite[Chap.
5]{do_carmo_riemannian_1992}. When $q_s$
is constant, i.e. $q_s=q$, we have the following connection between $J_t$ and
the differential $\tm\Exp$:
\begin{equation}
    \tm_{tv_0}\Exp_{q}tw
    =
    J_t
    \ .
    \label{eq:jacobi-dexp}
\end{equation}
Jacobi fields can equivalently be defined as solutions to an ODE that involves
the curvature endomorphism of the manifold \cite[Chap.
5]{do_carmo_riemannian_1992}.  However, the curvature
endomorphism is not easily computed when the manifold is represented implicitly,
and, therefore, the ODE is hard to use for computational applications in this
case. In the next section, we
numerically compute Jacobi fields by integrating the
differential of the system \eqref{sys:impl-geo}.
\begin{figure}[h]
    \begin{center}
      \includegraphics[width=0.50\columnwidth,trim=30 20 10 10,clip=true]{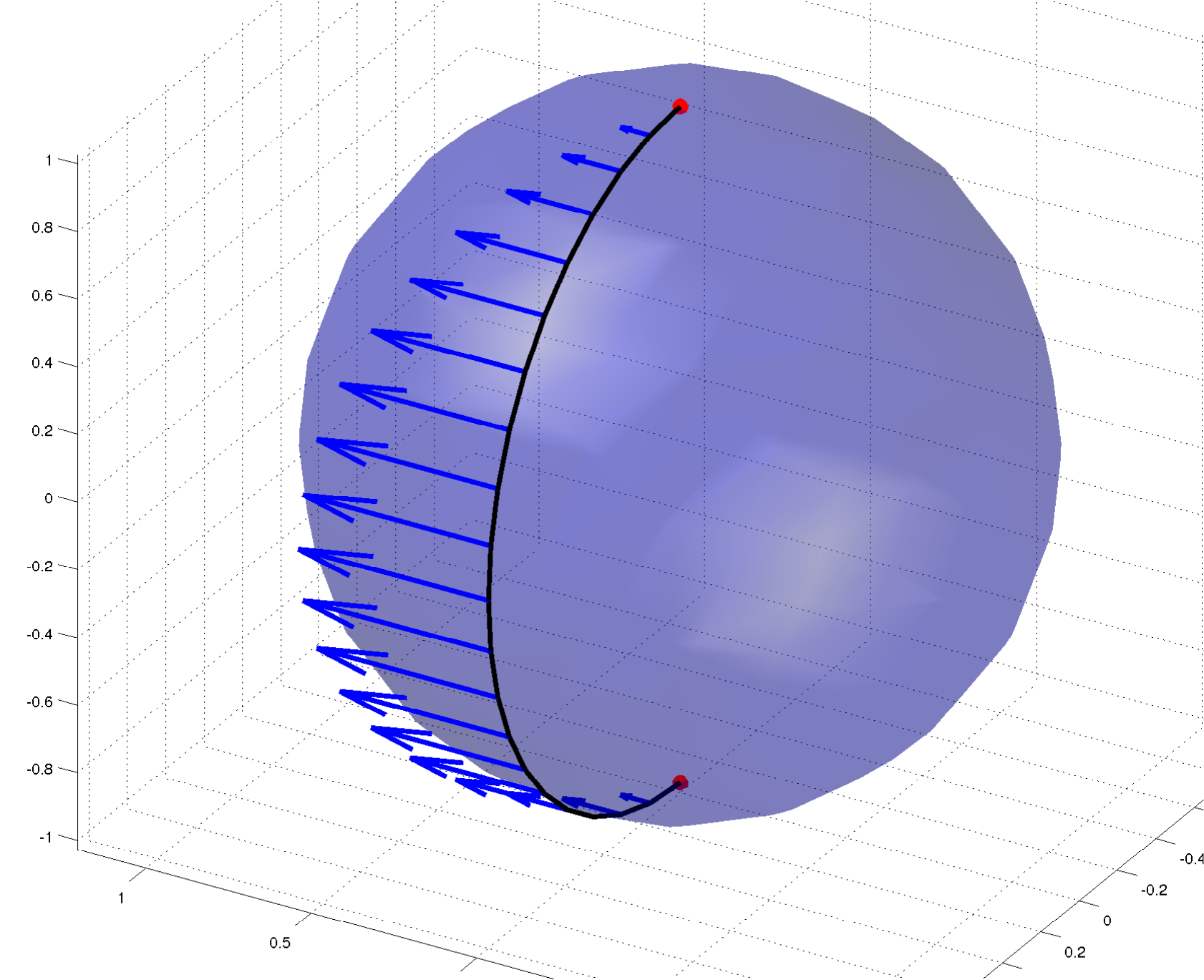}
    \end{center}
    \caption{The sphere $\SSS^2$ with a Jacobi field along a geodesic connecting
    the poles. Each pole is a conjugate point to the other since the non-zero Jacobi
    field vanishes. The injectivity radius is equal to the length of the
    geodesic, $\pi$.}
    \label{fig:sphere-jacobi}
\end{figure}

Jacobi fields can be used to retrieve various geometric information e.g.
sectional curvature. Let $J_t$ denote a Jacobi field along the geodesic $\alpha_t$ with
$J_0=0$ and derivative $w=\Df{J_0}{t}$. Assume the vectors $v_0=\dot{\alpha}_0$ and $w$ 
are orthonormal. These vectors define
a plane $\sigma=\Span\{v_0,w\}$ in $T_{\alpha_0}M$, and $K_{\alpha_0}(\sigma)$ denotes the sectional
curvature of the plane $\sigma$. Because $K_{\alpha_0}(\sigma)$ occurs in a Taylor expansion of the length $\|J_t\|$,
the sectional curvature can be estimated by
\begin{equation}
    K_{\alpha_0}(\sigma)
    \approx
    \frac{6}{t^3}(t-\|J(t)\|)
    \label{eq:sec-curve}
\end{equation}
for small $t$.
Furthermore, if $J_t$ is a non-zero Jacobi field with $J_0=0$ along a geodesic
$\alpha_t$ and, for some
$\tilde{t}>0$, also $J_{\tilde{t}}=0$ then $\alpha_{\tilde{t}}$ is called a
conjugate point to $\alpha_0$. This
can provide an upper bound on the injectivity radius of $M$, that, in
general terms, specifies the minimum length of non-minimizing
geodesics. Figure~\ref{fig:sphere-jacobi} illustrates the situation on the
sphere $\SSS^2$. We will explore both these points in the experiments section.

\subsection{Principal Geodesic Analysis}
\label{sec:pga-def}
Principal Component Analysis (PCA) is widely used to model the
variability of data in Euclidean spaces. The procedure provides linear
dimensionality reduction by
defining a sequence of linear subspaces maximizing the variance of the projection
of the data to the subspaces or, equivalently, minimizing the reconstruction errors. The $k$th 
subspace is spanned by an orthogonal
basis $\{v^1,\ldots,v^k\}$ of principal components $v^1,\ldots,v^k$, and
the $i$th principal component is defined recursively by
\begin{equation}
    v^i
    =
    \argmax_{\|v\|=1}
    \frac{1}{N}
    \sum_{j=1}^N
    \left(\ip{x_j,v}^2+\sum_{l=1}^{i-1} \ip{x_j,v^l}^2\right)
    \label{eq:pca-cost}
\end{equation}
when formulated as to maximize the variance of the projection of the dataset
$\{x_1,\ldots,x_N\}$ to the subspaces $\Span\{v^1,\ldots,v^{i-1}\}$.

PCA is dependent on the vector space structure of the Euclidean space and hence cannot be performed 
on manifold valued datasets. Principal Geodesic
Analysis was developed to overcome this limitation. 
PGA finds geodesic subspaces centered at point $\mu\in M$ with $\mu$ usually being
an intrinsic mean\footnote{
The notion of intrinsic mean goes back to Fr\'echet \cite{frechet_les_1948} 
and Karcher \cite{karcher_riemannian_1977}.
As in \cite{fletcher_principal_2004-1}, an intrinsic mean is here a minimizer of
$\argmin_{\mu\in M}\sum_{j=1}^Nd(\mu,x_j)^2$. Uniqueness issues are treated in 
\cite{karcher_riemannian_1977}.
}
of the dataset $\{x_1,\ldots,x_N\}$, $x_j\in M$.
The $k$th geodesic subspace $S_k$ of $T_\mu M$ is defined as
$\Exp_\mu(V_k)$ with $V_k=\Span\{v^1,\ldots,v^k\}$ 
being the span of the principal directions $v^1,\ldots,v^k$ defined recursively by
\begin{equation}
    \begin{split}
    &
    v^i
    =
    \argmax_{\|v\|=1,v\in V_{i-1}^\perp}
    \frac{1}{N}
    \sum_{j=1}^Nd(\mu,\pi_{S_v}(x_j))^2
    \ ,\\
    &S_v
    =
    \Exp_\mu(\Span\{V_{i-1},v\})
    \ .
    \end{split}
    \label{eq:pga-cost}
\end{equation}
The projection
$\pi_S(x)$ of a point $x\in M$
onto a geodesic subspace $S=\Exp_qV$ is
\begin{equation}
    \begin{split}
    \pi_S(x)
    &=
    \argmin_{y\in S}d(x,y)^2
    =
    \argmin_{y\in S}\|\Log_yx\|^2
    \\
    &=
    \Exp_q(\argmin_{w\in V}\|\Log_{\Exp_qw}x\|^2)
    \ .
    \end{split}
    \label{eq:proj-cost}
\end{equation}
The term being maximized in \eqref{eq:pga-cost} is the sample variance of the projected data, the expected value of the
squared distance to $\mu$, and PGA therefore extends PCA by finding \emph{geodesic subspaces} in 
which variance is maximized.\footnote{
A slightly different definition that uses 
one-dimensional subspaces and Lie group structure was introduced
in \cite{fletcher_statistics_2003}.
}

Since both optimization problems \eqref{eq:pga-cost} and
\eqref{eq:proj-cost} are difficult to optimize, PGA has traditionally
been computed using the orthogonal projection in the tangent space of $\mu$
to approximate the true projection. With this approximation,
equation \eqref{eq:pga-cost} simplifies to
\begin{equation*}
    v^i
    \approx
    \argmax_{\|v\|=1}
    \frac{1}{N}
    \sum_{j=1}^N
    \left(\ip{\Log_\mu x_j,v}^2+\sum_{l=1}^{i-1} \ip{\Log_\mu x_j,v^l}^2\right)
\end{equation*}
which is equivalent to \eqref{eq:pca-cost}, and, therefore,
the procedure amounts to performing regular PCA on the vectors $\Log_\mu x_j$.
We will refer to PGA with the approximation
as \emph{linearized} PGA, and, following \cite{said_exact_2007}, PGA as 
defined by \eqref{eq:pga-cost} will be referred to
as \emph{exact} PGA.\footnote{In \cite{said_exact_2007},
the fact that $\pi_S$ has a closed form
solution on the sphere $\SSS^3$ when $S$ is a one-dimensional geodesic subspace is used to 
iteratively compute
PGA with the \cite{fletcher_statistics_2003} definition.}
The ability to iteratively solve optimization problems over
geodesics that we will develop in the next sections will allow us to
optimize \eqref{eq:pga-cost} and hence numerically compute exact PGA.

In general, PGA might not be well-defined as the intrinsic mean might not be
unique and both existence and uniqueness may fail for the 
projections \eqref{eq:proj-cost} and the optimization problem \eqref{eq:pga-cost}. The convexity
bounds of Karcher \cite{karcher_riemannian_1977} ensures uniqueness of the mean
for sufficiently local data but setting up sufficient conditions to ensure
well-posedness of \eqref{eq:proj-cost} and \eqref{eq:pga-cost}
for general manifolds is difficult
because they depend on the global geometry of the manifold.

There is ongoing discussion of when principal components should be constrained to
pass the intrinsic mean as in PGA or if other types of means should be used, see 
\cite{huckemann_intrinsic_2010} with discussions. 
In Geodesic PCA \cite{huckemann_intrinsic_2010}, the principal geodesics do not necessarily pass the intrinsic
mean, and similar optimization that allows the PGA base point to move away from the
intrinsic mean can be carried
out with the optimization approach used in this paper.
PGA can also be modified by replacing maximization of sample variance by 
minimization of reconstruction error.
This alternate definition is not equivalent to the definition above, a fact that
again underlines the difference between the Euclidean and the curved
situation. We will illustrate differences between the formulations in the
experiments section but we mainly use the variance formulation \eqref{eq:pga-cost}.

\section{Optimization over Geodesics}
\label{sec:differentials}
Equation \eqref{eq:pga-cost} and \eqref{eq:proj-cost} defining 
PGA are examples of optimization problems over geodesics that in those cases are
represented by their starting point $\mu$ and initial velocity $v$.
More generally, we here consider problems
\begin{equation}
  \min_{(q,v)\in (M,T_q M)}F(\Exp_q v)
  \label{eq:geo-optim}
\end{equation}
where $F:M\rightarrow\RR$ is a function defining the cost of the geodesic
$\Exp_q tv$ here at time $t=1$.\footnote{Even more generally, $F$
can be a function of the entire curve $\Exp_q tv$, $t\in\RR$ instead of just
for the point $\Exp_q tv$, $t=1$ Note that for PGA, the initial velocity is in addition
constrained to subspaces of $T_q M$.}. In order to iteratively solve
optimization problems of the form \eqref{eq:geo-optim}, we will need 
derivatives of $\Exp_q v$ since
$d F(\Exp_q v)=d_yF\,d \Exp_q v$
with $y=\Exp_q v$. Thus, we wish to compute the differential of $\Exp_q v $
with respect to initial point $q$ and initial velocity $v$.
Since \eqref{eq:pga-cost} is a function of the
projection $\pi_S$ given by \eqref{eq:proj-cost}, we will later see that we need the second 
order differential of $\Exp$ as well.

Only in specific cases where explicit expressions for geodesics are available
can the above mention differentials be derived in closed form. Instead, for general
manifolds, the ODEs \eqref{eq:param-geo} and \eqref{sys:impl-geo} describing geodesics 
can be differentiated giving systems that can be numerically integrated to
provide the differentials.
This approach relies on the fact that sufficiently smooth initial value problems
(IVPs) are
differentiable with respect to their initial values, see e.g. \cite[Chap.
I.14]{hairer_solving_2008}.

We will here derive explicit expressions for IVPs describing the 
differential of the exponential map and Jacobi fields. In addition, we 
will differentiate the IVPs a second time. The concrete expressions will allow
the IVPs to be used for iterative optimization of problems on the form
\eqref{eq:geo-optim}. In particular, they will be used for the
exact PGA algorithm presented in the next section. 
The basic strategy is simple: we differentiate the geodesic systems of
Section~\ref{sec:geodesic-systems}. Though the resulting equations are
notationally complex, 
their derivation is in principle just repeated application of the
chain and product rules for differentiation.
MATLAB code for numerical integration of the systems is available at
\url{http://image.diku.dk/sommer}.

Since the geodesic equations \eqref{sys:impl-geo} contain the
generalized inverse of the Jacobian matrix $DF$, we will use the following
formula for derivatives of generalized inverses. When an $n\times m$ matrix $A_s$ depends
on a parameter $s$ and has full rank $n$, and if
its generalized inverse $A_s^\dagger$ is differentiable,
then the derivative $\df{(A_s^\dagger)}{s}$ is given by 
\begin{equation}
  \df{(A_s^\dagger)}{s}
  =
  -A_s^\dagger \big(\df{A_s}{s}\big)A_s^\dagger
  +\big(I-A_s^\dagger A_s\big)\big(\df{A_s}{s}^T\big)\big(A_s^\dagger\big)^TA_s^\dagger A_s \ .
  \label{eq:decell-expr}
\end{equation}
This result was derived in \cite{decell_derivative_1974,golub_differentiation_1973}
and \cite{hanson_extensions_1969} for the full-rank case.
We will apply \eqref{eq:decell-expr} with $A_s=D_{x_{t,s}}F$ when $x_{t,s}$ is an $s$ dependent
family of curves in the embedding space $\RR^m$ that are geodesics on $M$ and
when $t$ is fixed. To see 
that $D_{x_{t,s}}F^\dagger$ is
differentiable with respect to $s$ when $x_{t,s}$ depends smoothly on $s$, take a 
frame of the normal space to $M$ in a neighborhood of $x_{t,s}$, and note that 
$D_{x_{t,s}}F^\dagger$ is a composition of a invertible map onto the frame 
depending smoothly on $s$ and the frame itself. 

The explicit expressions for the differential equations are notationally
heavy. Therefore, we only
state the results here and postpone the actual derivation 
to Appendix~\ref{app:A}.
\begin{em}
  \vspace{0.5em}
  \\
    Let $M\subset\RR^m$ be defined as a regular zero
    level set of a $C^3$ map $F:\RR^m\rightarrow\RR^n$. Using the embedding, we
    identify curves in $M$ and vectors in $TM$ with curves and vectors in
    $\RR^m$. Let $x_t$ be a geodesic with $x_0=q$ and $\dot{x}_0=v$.
    The Jacobi field $J_t$ along $x_t$ with $J_0=u$ and
    $\Df{J_0}{t}=w$ can then be found as the $z$-part of the solution of the IVP
    \begin{equation}
        \begin{split}
            &
            \begin{pmatrix}
                \dot{y}_t\\
                \dot{z}_t
            \end{pmatrix}
            =
            F_{q,v}^I\left(t,
            \begin{pmatrix}
                y_t\\
                z_t
            \end{pmatrix}
            \right)
            \ ,\\
            &
            \begin{pmatrix}
                y_0\\
                z_0\\
            \end{pmatrix}
            =
            \begin{pmatrix}
                w\\
                u
            \end{pmatrix}
            \ ,
        \end{split}
        \label{sys:JI}
    \end{equation}
    with $F_{q,v}^I$ the map given in explicit form in Appendix~\ref{app:A}.
\end{em}\vspace{0.5em}\\
As previously noted, Jacobi fields can be described using an ODE incorporating
the curvature endomorphism in the parameterized case. We can, however, apply a
procedure similar to the implicit case and derive and IVP by differentiating
the geodesic system \eqref{eq:param-geo}. We will use the resulting IVP
\eqref{sys:JP} when working
with variations of geodesics in the parameterized case, see
Appendix~\ref{app:A}.

The systems \eqref{sys:JI} and \eqref{sys:JP} are both linear in the initial values $(w\ u)^T$ 
as expected of systems describing differentials. They are non-autonomous due to the
dependence on the position on the curve $x_t$.

Recall the equivalence \eqref{eq:jacobi-dexp} between Jacobi fields 
and $\tm\Exp$: if $(y_t,z_t)$ satisfy \eqref{sys:JI}
    (or \eqref{sys:JP}) with initial values $(w,0)^T$ then $\tm_v\Exp_q w$ is
    equal to $z_1$.
Therefore, we can compute the differential $\tm_v\Exp_q$ with respect to $v$ by
numerically integrating the system using a
basis $\{w^1,\ldots,w^\eta\}$ for the tangent space $T_qM$ at $q\in M$. With
initial conditions $(0,u)^T$ instead, we can similarly compute the derivative
with respect to the initial point $q$. Note 
that $\Exp_q\Log_q y=y$
implies that $\tm_y\Log_q=(\tm_{\Log_qy}\Exp_q)^{-1}$, a fact that allows the
computation of $\tm_y\Log_q$ as well.

Assuming the manifold is sufficiently smooth, we can differentiate the systems 
\eqref{sys:JI} and \eqref{sys:JP} once more and thereby obtain second order
information that we will need later. The main difficulty is
performing the algebra of the already complicated expressions for the systems,
and, for the implicit case, we will need second
order derivatives of the generalized inverses $D_{x_{t,s}}F^\dagger$.
For simplicity, we consider a families of geodesics with stationary initial
point. The derivations are
again postponed to Appendix~\ref{app:A}.
\begin{em}
  \vspace{0.5em}
  \\
    Let $M$ be of class $C^4$, and let $\alpha_{t,s}$ be a family of 
    geodesics. Assume $\param{x}:U\rightarrow M$ is a local parametrization
    containing $\alpha_{t,s}$, and let $x_{t,s}$
    be the curve in $U$ representing $\alpha_{t,s}$, i.e.
    $\param{x}^{-1}\circ\alpha_{t,s}=x_{t,s}$. Let $w\in T_qM$ with $\alpha_{0,s}=q$
    and $v_s=\dot{\alpha}_{0,s}$. Define $u=\df{v_0}{s}$, and let
    $V_{q,v_0,w,u}=\df{\left(d_{v_s}\Exp_qw\right)}{s}
    =\df{\Big(\df{\left(\Exp_qv_s+rw\right)}{r}\Big)}{s}$.
    Then, in coordinates defined by $\param{x}$, $V_{q,v_0,w,u}$ can be found as the 
    $r$-part of the solution of the IVP
    \begin{equation}
        \begin{split}
            &
            \begin{pmatrix}
                \dot{q}_t\\
                \dot{r}_t
            \end{pmatrix}
            =
            G_{q,v_0,w,u}^P\left(t,
            \begin{pmatrix}
                q_t\\
                r_t
            \end{pmatrix}
            \right)
            \ ,\\
            &
            \begin{pmatrix}
                q_0\\
                r_0\\
            \end{pmatrix}
            =
            \begin{pmatrix}
                0\\
                0
            \end{pmatrix}
            \ ,
        \end{split}
        \label{sys:JP2}
    \end{equation}
    with $G_{q,v_0,w,u}^P$ the map given in explicit form in Appendix~\ref{app:A}.

    Now, let instead $M\subset\RR^m$ be defined as a regular zero
    level set of a $C^4$ map $F:\RR^m\rightarrow\RR^n$.
    Then $V_{q,v_0,w,u}$ can be found as the $r$-part of the solution of the IVP
    \begin{equation}
        \begin{split}
            &
            \begin{pmatrix}
                \dot{q}_t\\
                \dot{r}_t
            \end{pmatrix}
            =
            G_{q,v_0,w,u}^I\left(t,
            \begin{pmatrix}
                q_t\\
                r_t
            \end{pmatrix}
            \right)
            \ ,\\
            &
            \begin{pmatrix}
                q_0\\
                r_0\\
            \end{pmatrix}
            =
            \begin{pmatrix}
                0\\
                0
            \end{pmatrix}
            \ ,
        \end{split}
        \label{sys:JI2}
    \end{equation}
    with $G_{q,v_0,w,u}^I$ the map given in explicit form in Appendix~\ref{app:A}.
\end{em}\vspace{0.5em}\\
We note that solutions to \eqref{sys:JP2} and \eqref{sys:JI2} depend linearly on $u$ even
though the systems themselves are not linear.

\subsection{Numerical Considerations}
The geodesic systems \eqref{eq:param-geo} and \eqref{sys:impl-geo} can in both 
the parametrized and implicit case be expressed
in Hamiltonian forms. In \cite{dedieu_symplectic_2005}, the authors use this
property along with symplectic numerical integrators to ensure 
the computed curves will be close to actual geodesics. This is possible since the Hamiltonian encodes
the Riemannian metric. The usefulness of pursuing a similar approach of 
expressing the differential systems in
Hamiltonian form and using symplectic integrators to preserve the Hamiltonians
is limited since there is no
direct interpretation of such Hamiltonians in contrast to the case for 
geodesic systems.

Along the same lines, we would like to use the preservation of quadratic
forms for symplectic integrators \cite{hairer_geometric_2002} to preserve
quadratic properties of the differential of the exponential map, e.g. the Gauss
Lemma \cite{do_carmo_riemannian_1992}. We are currently
investigating numerical schemes that could possibly ensure such stability.

\section{Exact Principal Geodesic Analysis}
\label{sec:pga}
As an example of how the IVPs describing differentials allow
optimizing over geodesics, we will provide algorithms that allow iterative
optimization of \eqref{eq:pga-cost} and that thus allow PGA as defined in
\cite{fletcher_principal_2004-1} to be computed without the traditional linear 
approximation.

Solving the optimization problem \eqref{eq:pga-cost} 
requires the ability to compute the projection $\pi_S$.
We start with the gradient needed for iteratively computing the projection before deriving
the gradient of the cost function of 
\eqref{eq:pga-cost}. Computing these gradients will
require the differentials over geodesic families
derived in Section~\ref{sec:differentials}.
Thereafter, we present the actual algorithms for solving the problems before
discussing convergence issues.

The optimization problems \eqref{eq:pga-cost} and \eqref{eq:proj-cost} are posed
in the tangent space of the manifold at the sample mean and the unit sphere
of that tangent space, respectively. These domains have relatively simple
geometry, and, therefore, the complexity of the problems is contained in the
cost functions. Because of this, we will not need optimizing
algorithms that are specialized for domains with complicated geometry.
For simplicity, we compute
gradients and present steepest descent algorithms but it is straightforward to
compute Jacobians instead and use more advanced optimization algorithms such as
Gauss-Newton or Levenberg-Marquardt.

The overall approach is similar 
to the approach used for computing exact PGA in \cite{said_exact_2007}. Our solution differs in that we are
able to compute $\pi_S$ and its differential without restricting to
the manifold SO(3) and in that we optimize the functional \eqref{eq:pga-cost} instead of
the cost function used in \cite{fletcher_statistics_2003} that involves one-dimensional 
subspaces.

\subsection{The Geodesic Subspace Projection}
We consider the projection $\pi_S(x)$ of a point $x\in
M$ on a geodesic subspace $S$. Assume $S$ is centered at $\mu\in M$, 
let $V$ be a $k$-dimensional subspace of $T_\mu M$
such that $S=\Exp_\mu V$, and define a residual function $R_{x,\mu}:V\rightarrow\RR$ 
by $w\mapsto\|\Log_{\Exp_\mu w}x\|^2$ that measures squared distances between $x$ and points
in $S$. Computing $\pi_S(x)$ by solving \eqref{eq:proj-cost} is then equivalent to
finding $w\in V$ minimizing $R_{x,\mu}$. To find the gradient of $R_{x,\mu}$, 
choose an orthonormal basis for $V$ 
and extend it to a basis for $T_\mu M$. Furthermore, let
$w_0\in V$ and choose an orthonormal basis for the tangent space $T_{\Exp_\mu w_0}M$. 
Karcher showed in \cite{karcher_riemannian_1977} that the gradient
$\grad^y\|\Log_yx\|^2$ equals $-2\Log_yx$, and, using this, we get the gradient of the
residual function as
\begin{equation}
    \nabla_{w_0}^{w\in V}R_{x,\mu}
        =
    -2(D_{w_0}\Exp_\mu)_{1,\ldots,k}^T
    (\Log_{\Exp_\mu w_0}x)
    \label{eq:proj-proj}
\end{equation}
with $(D_{w_0}\Exp_\mu)_{1,\ldots,k}$ denoting the first $k$ columns of
the matrix $D_{w_0}\Exp_\mu$ expressed using the chosen bases.\footnote{In
coordinates of the bases, the differential $d_{w_0}\Exp_\mu$ becomes a matrix
that we write $D_{w_0}\Exp_\mu$.
The notation $\nabla_{w_0}^{w\in V}$ denotes differentiation along the basis
elements of $V$. See Appendix~\ref{app:notation} for additional notation.}
This matrix can be computed using the IVPs \eqref{sys:JI} or \eqref{sys:JP}.

\subsection{The Differential of the Subspace Projection}
\label{sec:grad-proj}
In order to optimize \eqref{eq:pga-cost}, we will need to compute gradients of
the form
\begin{equation}
    \grad_{v_0}^{v\in V_{v_0}^\perp}
    d(\mu,\pi_{S_v}(x))^2
    \label{eq:grad-1}
\end{equation}
with $V_v=\Span\{v^1,\ldots,v^k,v\}$, $S_v=\Exp_\mu(V_v)$ and $\mu\in M$.\footnote{
Since $v$ in \eqref{eq:pga-cost} is restricted to the unit sphere, we will not
need the gradient in the direction of $v_0$, and, therefore, we find the gradient
in the subspace $V_{v_0}^\perp$ instead of in the larger space
$\Span\{v^1,\ldots,v^k\}^\perp$.
As noted in Section~\ref{sec:pga-def}, the optimization approach presented here can be extended to
include optimization of the base point $\mu$ as well. Here, we use a fixed base
point that for PGA is an intrinsic mean of a data set.
}
This will involve the differential of $\pi_{S_v}(x)$ with respect to $v$. 
Since $\pi_{S_v}(x)$ is defined as a minimizer of \eqref{eq:proj-cost}, 
its differential cannot be obtained just by applying the chain and product rules.
Instead, we use the implicit function theorem to define a map $\Psi$ that equals 
$\pi_{S_v}(x)$ around a neighborhood of $v$ in $T_\mu M$. We then derive the
differential of $\Psi$.

For the result below, we extend the domain of residual function $R_{x,\mu}$
defined above from $V$ to the entire tangent space $T_\mu M$.
We will a choose basis for $T_\mu M$, and we let $\Hess(R_{x,\mu})$
denote the Hessian matrix of $R_{x,\mu}$ with respect to the basis. Similarly,
we will choose a basis for $V_{v_0}$, and
we let $\Hess(R_{x,\mu}|_{V_{v_0}})$ denote the Hessian matrix of $R_{x,\mu}$
restricted to $V_{v_0}$ with respect to this basis. Using this notation, we 
get the following result for the derivative of the projection $\pi_{S_v}(x)$:
\begin{proposition}
    Let $\{v^1,\ldots,v^k\}$ be an orthonormal basis for a subspace $V\subset T_\mu M$. 
    For each $v\in V^\perp$, let $V_v$ be the subspace $\Span\{V,v\}$, and
    let $S_v=\Exp_\mu V_v$ be the corresponding geodesic subspace. 
    Fix $v_0\in V^\perp$ and define $w_0=\Log_\mu\pi_{S_{v_0}}(x)$ for an
    $x\in M$. Suppose the matrix $\Hess_{v_0}(R_{x,\mu}|_{V_{v_0}})$ has full rank $k+1$. Extend the
    orthonormal basis $\{v^1,\ldots,v^k,v_0/\|v_0\|\}$
    for $V_{v_0}$ to an orthonormal basis for 
    $T_\mu M$. Then
    \begin{equation}
        \begin{split}
        &D^{v\in V_{v_0}^\perp}_{v_0}\pi_{S_v}(x)
        =
        -(D_{w_0}\Exp_\mu)\bar{v}_{x,\mu,v_0,S_{v_0}}
        \left( 
        \nabla_{w_0}^{w\in V_{v_0}^\perp} R_{x,\mu}
        \right)^T
        \\
        &\qquad\qquad\qquad\quad
        +w_0^{k+1}(D_{w_0}\Exp_\mu)E_{x,\mu,v_0,S_{v_0}}
        \ .
        \end{split}
        \label{eq:grad-proj}
    \end{equation}
    The coordinates of the vector $\bar{v}_{x,\mu,v_0,S_{v_0}}$
    in the basis for $V_{v_0}$ are contained in the $(k+1)$st column of 
    the matrix $\Hess_{v_0}(R_{x,\mu}|_{V_{v_0}})^{-1}$, the scalar $w_0^{k+1}$
    is the $(k+1)$st coordinate of $w_0$ in the basis, and
    $E_{x,\mu,v_0,S_{v_0}}$ is the matrix
    \begin{equation*}
        \begin{pmatrix}
            -\Hess_{w_0}\left(R_{x,\mu}|_{V_{v_0}}\right)^{-1}B_{w_0,v_0} \\
            I_{\eta-(k+1)}
        \end{pmatrix}
    \end{equation*}
    with $B_{w_0,v_0}$ the last $\eta-(k+1)$ columns of 
    the matrix
    $(\Hess_{w_0}\left(R_{x,\mu}\right)(V\ v_0))^T$
    and
    $I_{\eta-(k+1)}$
    the identity matrix.
    \label{prop:grad-proj}
\end{proposition}
Before proving the result, we discuss its use for computing the gradient
\eqref{eq:grad-1}. 
The assumption that the Hessian of the restricted residual $R_{x,\mu}|_{V_{v_0}}$ 
must have full rank is discussed below.

Because $d(\mu,\pi_{S_v}(x))^2=\|\Log_\mu\pi_{S_v}(x)\|^2$, 
we have
\begin{equation}
    \begin{split}
        \nabla_{v_0}^{v\in V_{v_0}^\perp}d(\mu,\pi_{S_v}(x))^2
        =
        2\left((D_{\pi_{S_{v_0}}(x)}\Log_\mu)(D_{v_0}^{v\in
        V_{v_0}^\perp}\pi_{S_v}(x))\right)^T
        (\Log_\mu\pi_{S_{v_0}}(x))
        \ ,
    \end{split}
    \label{eq:full-expr}
\end{equation}
which, combined with \eqref{eq:grad-proj}, gives \eqref{eq:grad-1}.
In order to compute the right hand side of \eqref{eq:grad-proj}, it is necessary to
compute parts of the Hessian of the non-restricted residual $R_{x,\mu}$.
For doing this, we will use the alternative formulation 
$R_{x,\mu}(w)=\|\Log_x\Exp_\mu w\|^2$ for the residual function.
With $w_0,v\in T_\mu M$ let $y=\Exp_\mu w_0$.
Working in the chosen orthonormal basis, we have
\begin{equation*}
    \nabla_{w_0} R_{x,\mu}
    =2
    \left(\left(D_y\Log_x\right)D_{w_0}\Exp_\mu\right)^T
    \Log_xy
    \ .
\end{equation*}
and hence
\begin{equation}
    \begin{split}
        &\df{\left(\nabla_{w_0+vs} R_{x,\mu}\right)}{s}|_{s=0}
        \\
        &\quad
        =
        2\left(
        \df{\left(D_{\Exp_\mu (w_0+sv)}\Log_x\right)}{s}|_{s=0}
        \left(D_{w_0}\Exp_\mu\right)
        \right)^T
        \Log_xy
        \\
        &
        \qquad
        +
        2\left(
        \left(D_{y}\Log_x\right)
        \df{\left(D_{w_0+vs}\Exp_\mu\right)}{s}|_{s=0}
        \right)^T
        \Log_xy
        \\
        &
        \qquad
        +
        2\left(
        \left(D_{y}\Log_x\right)
        \left(D_{w_0}\Exp_\mu\right)
        \right)^T
        \df{\left(\Log_x\Exp_\mu (w_0+sv)\right)}{s}|_{s=0}
        \ .
    \end{split}
    \label{eq:hessian}
\end{equation}
Note that
\begin{align*}
    \df{\left(\Log_x\Exp_\mu (w_0+sv)\right)}{s}|_{s=0}
    =
    \left(D_y\Log_x\right)
    \left(D_{w_0}\Exp_\mu\right)
    v
    \ .
\end{align*}
Using that $\df{(A_s^{-1})}{s}=A_s^{-1}(\df{A_s}{s})A_s^{-1}$ for a time dependent,
invertible matrix $A_s$\footnote{
See \cite[Eq. (2)]{decell_derivative_1974}.
} and the fact that
$\Exp_x\Log_x z=z$ for all $z$, we get
\begin{align*}
    &\df{\left(D_{\Exp_\mu (w_0+sv)} \Log_x\right)}{s}|_{s=0}
    =
    \df{\left(D_{\Log_x(\Exp_\mu w_0+sv)}\Exp_x\right)^{-1}}{s}|_{s=0}
    \\
    &\quad
    =
    -\left(D_y\Log_x\right)
    \df{\left(D_{\Log_x(\Exp_\mu w_0+sv)}\Exp_x\right)}{s}|_{s=0}
    \left(D_y\Log_x\right)
    \ .
\end{align*}
The middle term of this product and the term 
$\df{\left(D_{w_0+sv}\Exp_\mu\right)}{s}|_{s=0}$ in \eqref{eq:hessian}
can be computed using the IVPs \eqref{sys:JP2},\eqref{sys:JI2} discussed in
Section~\ref{sec:differentials}.

\begin{proof}[Proposition~\ref{prop:grad-proj}]
Extend the basis
$\{v^1,\ldots,v^k,v_0/\|v_0\|\}$ for $V_{v_0}$ to an orthonormal basis for 
$T_\mu M$. The argument is not dependent on this choice of basis, but it will 
make the reasoning and notation easier. Let $S\subset T_\mu M\times V^\perp$ be an open neighborhood of $(w_0,v_0)$ and
define the map $F_V:S\rightarrow\RR^\eta$ by
\begin{equation*}
    F_V(w,v)
    =
    \begin{pmatrix}
        \nabla_w R_{x,\mu}\cdot v^1 \\
        \vdots \\
        \nabla_w R_{x,\mu}\cdot v^k \\
        \nabla_w R_{x,\mu}\cdot v \\
        w\cdot u^1(v) \\
        \vdots \\
        w\cdot u^{\eta-k-1}(v) 
    \end{pmatrix}
    =
    \begin{pmatrix}
        \Big(V\ v\Big)^T\nabla_w R_{x,\mu} \\
        U_v^Tw \\
    \end{pmatrix}
\end{equation*}
with the vectors $u^1(v),\ldots,u^{\eta-(k+1)}(v)$ constituting an orthonormal basis for
$V_v^\perp$ for each $v$ and with $(V\ v)$ and $U_v$ denoting the matrices having $v^i,v$
and $u^i(v)$ in the columns, respectively. 
Since $\ip{\nabla_{w_0}R_{x,\mu},v}=d_{w_0}R_{x,\mu}(v)=0$ for all
$v\in V_{v_0}$ because $w_0$ is a minimizer for $R_{x,\mu}$ among
vectors in in $V_{v_0}$, we see that 
$F_V(w_0,v_0)$ vanishes. Therefore, if $D^w_{(w_0,v_0)}F_V$ is 
non-singular, the implicit function theorem asserts the existence of a map $\Psi$ 
from a neighborhood 
of $v_0$ to $T_\mu M$ with the property that $F_V(\Psi(v),v)=0$ for all $v$ in
the neighborhood. We then compute
\begin{equation*}
    0
    =D_{v_0}F_V(\Psi(v),v)
    =
    \left(D^w_{(w_0,v_0)}F_V\right)\left(D_{v_0}\Psi(v)\right)
    +\left(D^v_{(w_0,v_0)}F_V\right)
\end{equation*}
and hence
\begin{equation}
    D^{v\in V_{v_0}^\perp}_{v_0}\Psi(v)
    =
    -\left(D^w_{(w_0,v_0)}F_V\right)^{-1}
    \left(D^{v\in V_{v_0}^\perp}_{(w_0,v_0)}F_V\right)
    \ .
    \label{eq:dpsi}
\end{equation}
For the differentials on the right hand side of \eqref{eq:dpsi}, we have
\begin{equation*}
    D^{v\in V_{v_0}^\perp}_{(w_0,v_0)}F_V
    =
    \Big(
        0\ 
        \cdots\ 
        0\ \ 
        \nabla_{w_0}^{w\in V_{v_0}^\perp}R_{x,\mu}\ \ 
        D^{v\in V_{v_0}^\perp}_{(w_0,v_0)}\left(w_0^TU_v\right)
       \Big)^T
\end{equation*}
 and
\begin{equation}
    \begin{split}
    D^w_{(w_0,v_0)}F_V
    &=
    \begin{pmatrix}
        \Big(V\ v_0\Big)^T\tmp{w}_{w_0}\left(\nabla_w R_{x,\mu}\right) \\
        U_{v_0}^T \\
    \end{pmatrix}
    =
    \begin{pmatrix}
        \left(\Hess_{w_0}\left(R_{x,\mu}\right) \Big(V\ v_0\Big)\right)^T
         \\ 
        U_{v_0}^T \\
    \end{pmatrix}
     \ .
    \end{split}
    \label{eq:dwFv-splitup}
\end{equation}
With the choice of basis, the above matrix is block triangular,
\begin{equation}
    D^w_{(w_0,v_0)}F_V
    =
    \begin{pmatrix}
        A_{w_0,v_0} & B_{w_0,v_0}\\
        0 & C_{w_0,v_0}
    \end{pmatrix}
        \ ,
    \label{eq:blockt}
\end{equation}
with $A_{w_0,v_0}$ equal to $\Hess_{w_0}(R_{x,\mu}|_{V_{v_0}})$.  The
requirement that $D^w_{(w_0,v_0)}F_V$ is non-singular is 
fulfilled, because $\Hess_{w_0}(R_{x,\mu}|_{V_{v_0}})$ has rank $k+1$ by assumption
and $U_{v_0}$ has rank $\eta-(k+1)$.

Since the first $k$ rows of $D^{v\in V_{v_0}^\perp}_{(w_0,v_0)}F_V$ are zero, we need only 
the last $\eta-k$ columns of $(D^w_{(w_0,v_0)}F_V)^{-1}$ in order to
compute \eqref{eq:dpsi}. The vector $\bar{v}_{x,\mu,v_0,S_{v_0}}$ as defined in
the statement of the theorem is equal to the $(k+1)$st column. 
Let $E_{x,\mu,v_0,S_{v_0}}$ be the matrix consisting of the remaining $\eta-(k+1)$ 
columns. Using the form \eqref{eq:blockt}, we have
\begin{equation*}
    E_{x,\mu,v_0,S_{v_0}}
    =
    \begin{pmatrix}
        -\Hess_{w_0}\left(R_{x,\mu}|_{V_{v_0}}\right)^{-1}B_{w_0,v_0} C_{w_0,v_0}^{-1} \\
        C_{w_0,v_0}^{-1}
    \end{pmatrix}
    \ .
\end{equation*}

Assume $\{u^1,\ldots,u^j\}$ is chosen such that $\{u^1(v_0),\ldots,u^j(v_0)\}$
equals the previously chosen basis for $V_{v_0}^\perp$. With this assumption,
$C_{w_0,v_0}$ is the identity matrix $I_{\eta-(k+1)}$.
In addition, let $w_0^{k+1}$ denote the $(k+1)$st component of $w_0$, 
that is, the projection of
$w_0$ onto $v_0/\|v_0\|$. Since $w_0\in V_{v_0}$ and
by choice of $U_v$, \lemref{orth-deriv} (see Appendix~\ref{app:B}) gives
\begin{equation*}
    D^{v\in V_{v_0}^\perp}_{(w_0,v_0)}\left(U_v^Tw\right)
    =
    w_0^{k+1}D^{v\in V_{v_0}^\perp}_{(w_0,v_0)}\left(U_v^T\tfrac{v_0}{\|v_0\|}\right)
    =
    -w_0^{k+1}I_{\eta-(k+1)}
    \ .
\end{equation*}
Therefore,
\begin{equation}
    D^{v\in V_{v_0}^\perp}_{(w_0,v_0)}F_V
    =
    \Big(
        0\ 
        \cdots\ 
        0\ \ 
        \nabla_{w_0}^{w\in V_{v_0}^\perp} R_{x,\mu}\ \ 
        -w_0^{k+1}I_{\eta-(k+1)}
       \Big)^T
        \ .
        \label{eq:Fv}
\end{equation}
Note, in particular, that $D^{v\in V_{v_0}^\perp}_{(w_0,v_0)}F_V$ is independent
on the actual choice of bases $U_v$. Combining \eqref{eq:dpsi}, \eqref{eq:Fv},
and the fact that
$\bar{v}_{x,\mu,v_0,S_{v_0}}$ and $E_{x,\mu,v_0,S_{v_0}}$ constitute the needed
columns of $(D^w_{(w_0,v_0)}F_V)^{-1}$, we get 
\begin{equation*}
    D^{v\in V_{v_0}^\perp}_{v_0}\Psi(v)=-\bar{v}_{x,\mu,v_0,S_{v_0}}(\nabla_{w_0}^{w\in V_{v_0}^\perp}
    R_{x,\mu})^T
    +w_0^{k+1}E_{x,\mu,v_0,S_{v_0}}
    \ .
\end{equation*}
Because $\Exp_\mu\Psi(v)=\pi_{S_v}(x)$, this provides \eqref{eq:grad-proj}.

\end{proof}

\subsection{Exact PGA Algorithm}
The gradients of the cost functions enable us to 
iteratively solve the optimization problems
\eqref{eq:pga-cost} and \eqref{eq:proj-cost}.
We let $\mu$ be an intrinsic mean of a dataset $\{x_1,\ldots,x_N\}$, 
$x_i\in M$.
The algorithms listed below are essentially steepest descent methods but,
as previously noted, Jacobian-based optimization algorithms can be employed as well.

Algorithm~\ref{alg:proj} for computing $\pi_S(x)$ updates $w\in V$
instead of the actual point $y\in S$ that we are interested in. The vector $w$ is
related to $y$ by $y=\Exp_\mu w$.
\begin{algorithm}
\caption{Calculate $\pi_S(x)$}
\label{alg:proj}
\begin{algorithmic}
    \REQUIRE $x\in M$, $S=\Exp_\mu V$ geodesic subspace.
    \STATE $w \Leftarrow \mbox{orthogonal projection of }\Log_\mu x \mbox{ onto
    }V$
    \COMMENT{initial guess} 
  \REPEAT 
  \STATE $y \Leftarrow \Exp_\mu w$
    \COMMENT{vector to point}
  \STATE $g \Leftarrow -2(D_{w_0}\Exp_\mu)_{1,\ldots,k}^T\Log_yx$ 
    \COMMENT{gradient} 
  \STATE $\tilde{w} \Leftarrow w$ 
    \COMMENT{previous $w$}
  \STATE $w \Leftarrow w-g$ 
    \COMMENT{update $w$}
  \UNTIL{$\|\tilde{w}-w\|$ is sufficiently small}.
\end{algorithmic}
\end{algorithm}

The algorithm for solving \eqref{eq:pga-cost}
is listed in Algorithm~\ref{alg:pga}. 
Since $v$ in \eqref{eq:pga-cost} is required to be on the unit sphere, 
the optimization will take place on a manifold,
and a natural approach to compute iteration updates will use the exponential
map of the sphere. Yet, because of the symmetric geometry of the sphere, we 
approximate this using the simpler method
of adding the gradient to the previous guess and normalizing.
When computing the $(k+1)st$ principal direction, we 
choose the initial guess as the first regular PCA vector of the data projected
to $V_k^\perp$ in $T_\mu M$. See Figure~\ref{fig:alg} for
an illustration of an iteration of the algorithm.
\begin{algorithm}
    \caption{Calculate the $(k+1)st$ principal direction of \eqref{eq:pga-cost}.}
\label{alg:pga}
\begin{algorithmic}
    \REQUIRE $\mu,x_1,\ldots,x_N\in M$, $\{v^1,\ldots,v^k\}$ orthogonal basis for
    $V_k\subset T_\mu M$.
  \STATE $v \Leftarrow $ first PCA vector of $\{x_j\}$ projected first to
  $T_\mu M$\\\hspace{2.1em}using $\Log_\mu$ and then to $V_k^\perp$
    \COMMENT{initial guess} 
  \REPEAT 
  \STATE $g_j \Leftarrow \nabla_{v}^{v\in V_{v}^\perp}d(\mu,\pi_{S_v}(x_j))^2$
    \COMMENT{for each $j$ using \eqref{eq:full-expr}}
    \STATE $g \Leftarrow \frac{1}{N}\sum_{j=1}^N g_j
$ 
    \COMMENT{gradient} 
  \STATE $\tilde{v} \Leftarrow v$ 
    \COMMENT{previous $v$}
  \STATE $v \Leftarrow v+g$ 
    \COMMENT{update $v$}
  \STATE $v \Leftarrow v/\|v\|$ 
    \COMMENT{normalize}
  \UNTIL{$\|\tilde{v}-v\|$ is sufficiently small}.
\end{algorithmic}
\end{algorithm}
\begin{figure}[h]
    \begin{center}
      \includegraphics[width=0.5\columnwidth,trim=0 40 0 0,clip=true]{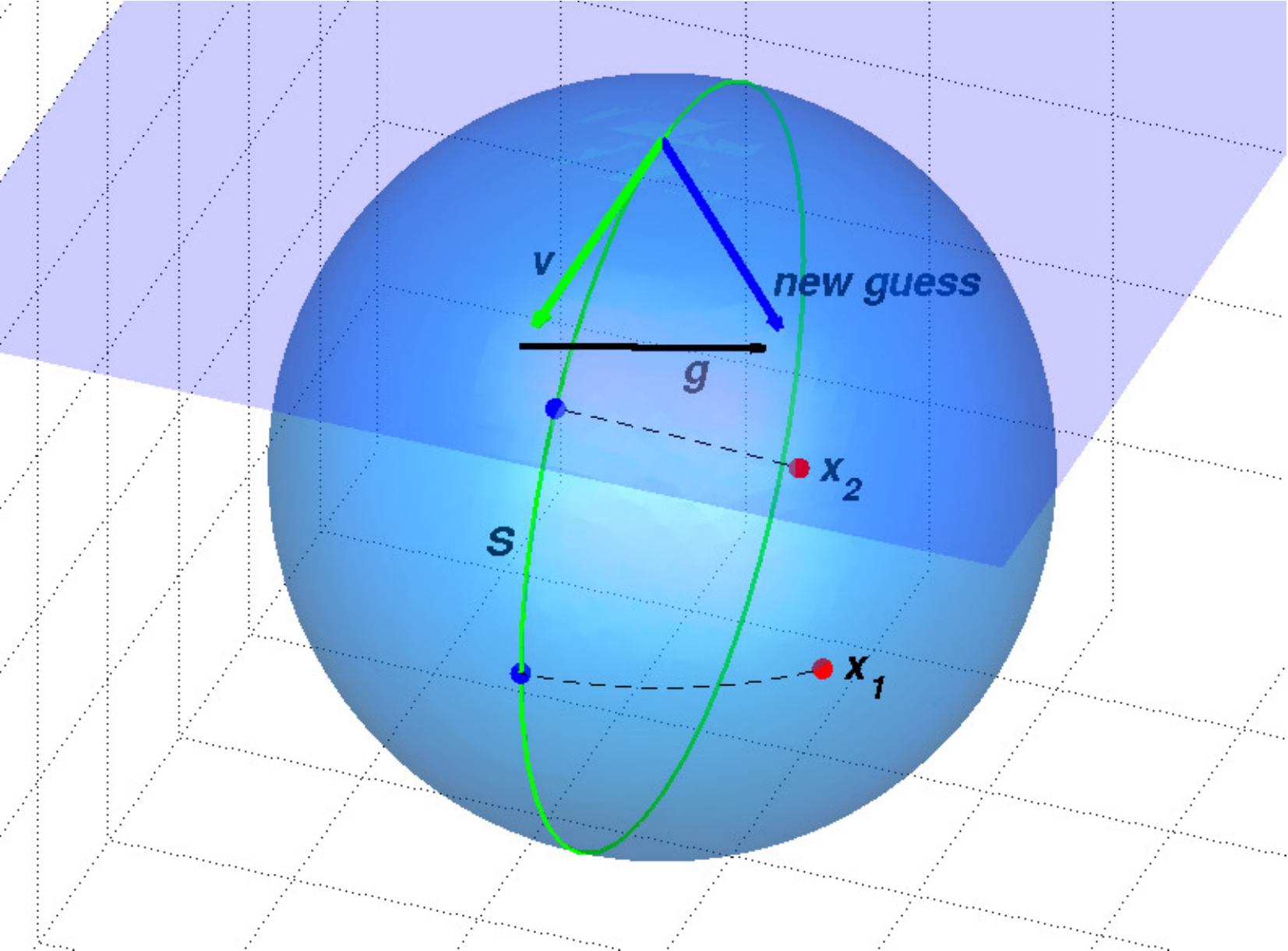}
    \end{center}
    \caption{An iteration of Algorithm~\ref{alg:pga}. The figure shows data points $x_1$ and
    $x_2$ (red points) with projections (blue points) to the geodesic subspace $S$ (green
    line). The vector $v$ defining $S$ is updated to the new guess by adding the
  gradient $g$ and normalizing.}
    \label{fig:alg}
\end{figure}

\subsection{Assumptions and Convergence}
As discussed in section \ref{sec:pga-def}, because a uniqueness and existence of
both the intrinsic mean and optima for \eqref{eq:pga-cost} may fail, the PGA
problem may not be well defined in itself. The uniqueness of the mean can be
obtained by assuming the data is sufficiently concentrated depending on the
curvature, see \cite{karcher_riemannian_1977}.

The curvature of the manifold may make the optimization problems non-convex,
and convergence to a global optimum is therefore only
ensured under the assumption that the problems
\eqref{eq:pga-cost} and \eqref{eq:proj-cost} are convex or that no local minima
exist.  Giving criteria for convexity or non-existence of local optima for
general manifolds and data sets is 
difficult because of the dependence on the global geometry of the manifolds.

The rank assumption on the Hessian used in Proposition~\ref{prop:grad-proj} is
equivalent to the residual $R_{x,\mu}$ having only non-degenerate critical points when
restricted to $V_{v_0}$. 
It is shown in \cite{karcher_riemannian_1977} that $R_{x,\mu}$ is convex at points sufficiently
close to $x$ and the assumption is therefore satisfied in such cases. In
particular, this is satisfied if Algorithm~\ref{alg:pga} is initialized with
subspaces that provide a good approximation to the data.

\section{Experiments}
\label{sec:experiments}
We will use the optimization strategy and the developed algorithm for exact PGA
to illustrate the
differences between exact and linearized PGA. Furthermore, we will estimate sectional
curvatures and compute injectivity radius bounds.
Even though the algorithms are not limited to low dimensional data, we aim at visualizing the
results and we will therefore provide examples with synthetic data on low
dimensional manifolds. The setup allows exploring the connection between
the geometry and curvature of the manifolds and the exact PGA result, and we
will show how the variance and residual formulation can provide fundamentally
different results.
For a comparison between the methods on high dimensional 
manifolds modeling real-life data, we refer the reader to
\cite{sommer_manifold_2010} where datasets of
human vertebrae X-rays and motion capture data are treated.

The PGA algorithm is implemented in Matlab using 
Runge-Kutta ODE solvers. For the logarithm map, we use the shooting algorithm developed
in \cite{sommer_bicycle_2009}. All tolerances used for the integration and
logarithm calculations are set at or lower than an order of magnitude of the 
precision used for the displayed results. Intrinsic means are computed by
iteratively minimizing variance using the gradient $\grad^y\|\Log_yx\|^2=-2\Log_yx$ (see
\cite{karcher_riemannian_1977}). The code used for the experiments
is available at \url{http://image.diku.dk/sommer}.

\ \\
We first consider surfaces embedded in $\RR^3$ and defined by the equation
\begin{equation*}
    S_c=\{(x_1,x_2,x_3)|cx_1^2+x_2^2+x_3^2=1\}
\end{equation*}
for different values of the scalar $c$. For $c>0$, $S_c$ is an ellipsoid
and it is equal to $\SSS^2$ in the case $c=1$. The surface $S_0$ is a cylinder and, for
$c<0$, $S_c$ is hyperboloid. Consider the point $p=(0,0,1)$ and note that $p\in S_c$ for all $c$. 
The curvature of $S_c$ at $p$ is equal to $c$. Note in particular that for the
cylinder case the curvature is zero; the cylinder locally has the geometry
of the plane $\RR^2$ even though it informally seems to curve.

We evenly distribute 20 points along two straight lines
through the origin of the tangent space $T_pS_c$, project the points 
from $T_pS_c$ to the surface $S_c$, and perform linearized and exact PGA. 
Figure~\ref{fig:low} illustrates the
situation in $T_pS_{-1}$ and on $S_{-1}$ embedded in $\RR^3$, respectively.
The lines are chosen in order to ensure the points are spread over areas of the
surface with different geometry. This choice is made to illustrate the influence
of the curvature;
a more realistic example with points sampled from a Gaussian will be provided
below.
\begin{figure}[h]
    \begin{center}
      \subfigure[$T_pS_{-1}$ with sampled points and first principal components
      (blue exact PGA, green linearized PGA).\label{fig:low1}]{\includegraphics[width=0.45\columnwidth,trim=30 20 10 10,clip=true]{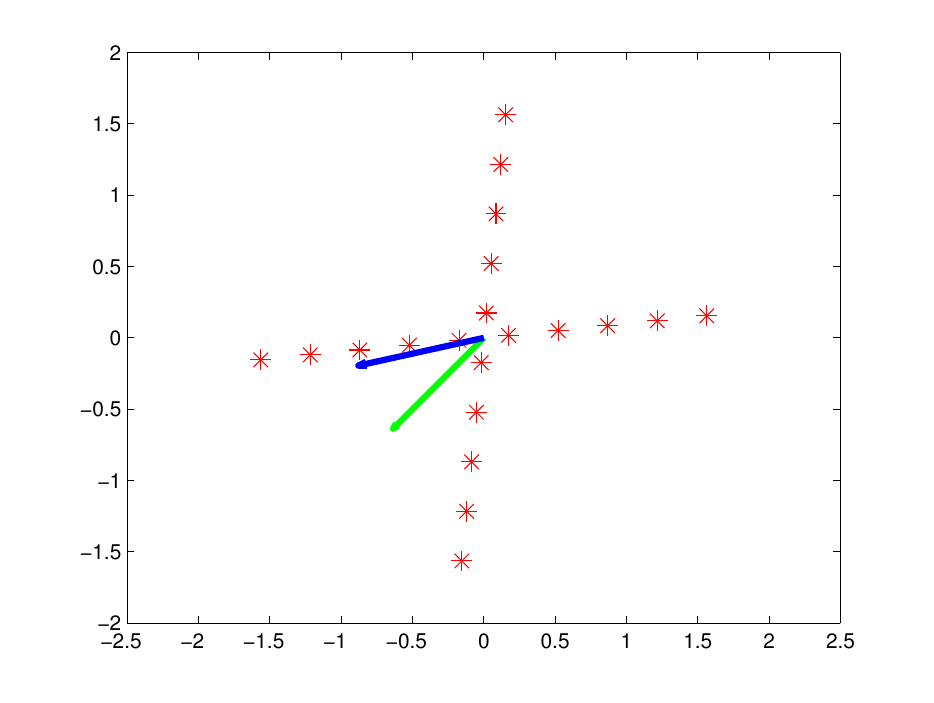}}
      \subfigure[$S_{-1}$ with projected points and first principal components (blue
      exact PGA \eqref{eq:pga-cost}, green linearized PGA).\label{fig:low2}]{\includegraphics[width=0.5\columnwidth,trim=40 30 70 50,clip=true]{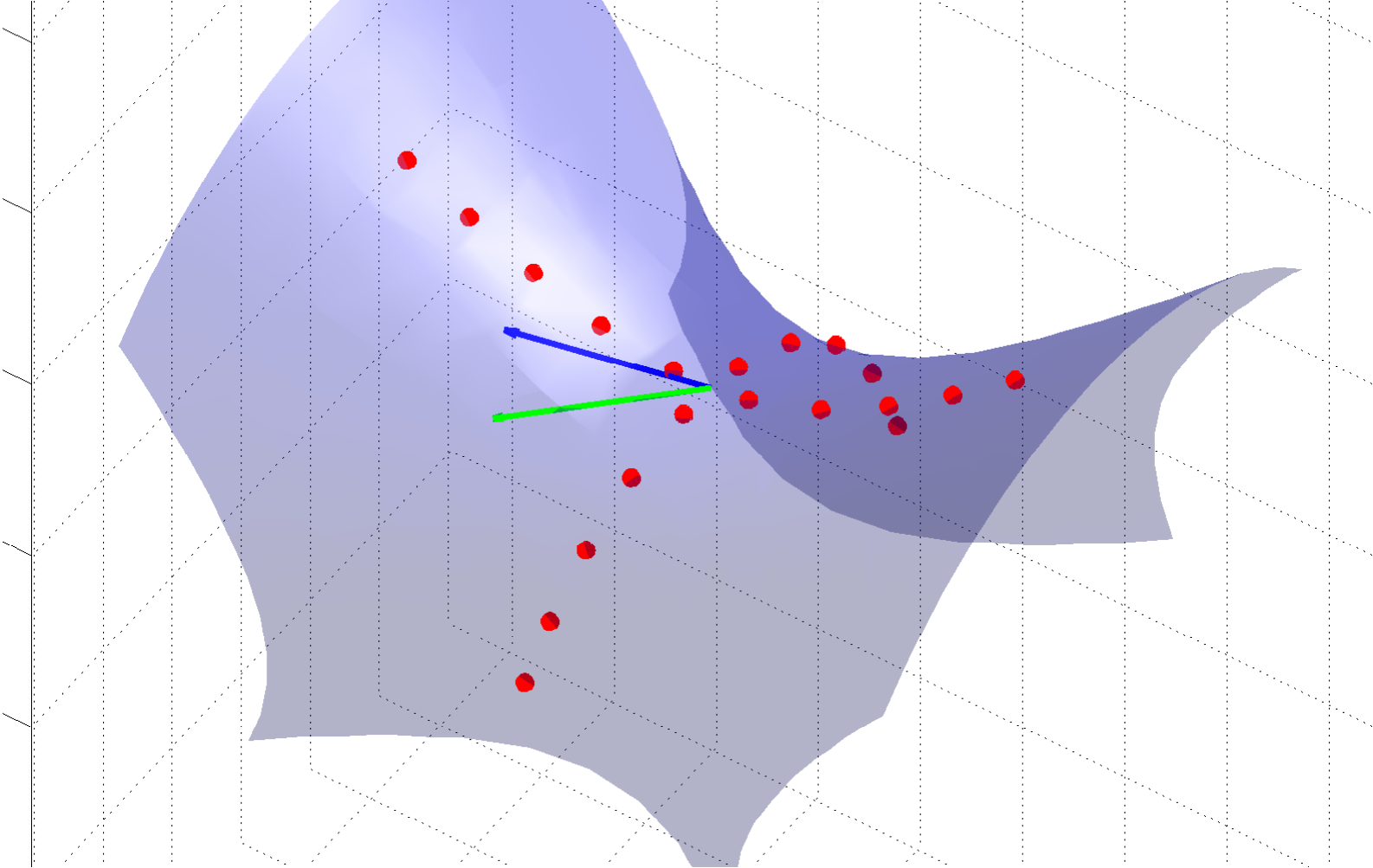}}
    \end{center}
    \caption{The tangent space $T_pS_{-1}$ and the manifold $S_{-1}$ with sample
  points.}
    \label{fig:low}
\end{figure}

Since linearized PCA amounts to Euclidean PCA in $T_pS_c$, the first principal
direction found using linearized PGA divides the angle between the lines for all $c$. 
In contrast to this, the variance and the
first principal direction found using exact PGA are dependent on $c$.
Table~\ref{table:low1} shows the angle between the principal directions found
using the two methods, the variances and variance differences for different
values of $c$. 
\begin{table}[ht]
    \scriptsize
    \setlength{\tabcolsep}{5pt}
\begin{center}
\begin{tabular}{lrrrrrrrrrrr}
  \hline
  \bf{c}: & 
  \multicolumn{1}{c}{\bf{1}} & 
  \multicolumn{1}{c}{\bf{0.5}} &
  \multicolumn{1}{c}{\bf{0}} &
  \multicolumn{1}{c}{\bf{-0.5}} &
  \multicolumn{1}{c}{\bf{-1}} &
  \multicolumn{1}{c}{\bf{-1.5}} &
  \multicolumn{1}{c}{\bf{-2}} & 
  \multicolumn{1}{c}{\bf{-3}} &
  \multicolumn{1}{c}{\bf{-4}} &
  \multicolumn{1}{c}{\bf{-5}}\\
  \hline
  angle $(^\circ)$:     & 0.0   & 0.1   & 0.0   & 22.3  & 29.2  & 31.5  & 32.6 & 33.8  & 34.2   & 34.5 \\
  linearized var.:         & 0.899 & 0.785 & 0.601 & 0.504 & 0.459 & 0.435 & 0.423 & 0.413 & 0.413 & 0.417 \\
  exact var.:           & 0.899 & 0.785 & 0.601 & 0.525 & 0.517 & 0.512 & 0.510 & 0.508 & 0.507 & 0.506 \\
  difference:           & 0.000 & 0.000 & 0.000 & 0.012 & 0.058 & 0.077 & 0.087 & 0.095 & 0.094 & 0.089 \\
  difference (\%):      & 0.0   & 0.0   & 0.0   & 4.2   & 12.5  & 17.6  & 20.6 & 23.0  & 22.7   & 21.4 \\
   \hline
\end{tabular}
\caption{Differences between methods for different values of $c$.}
\label{table:low1}
\end{center}
\end{table}

Let us give a brief explanation of the result. The symmetry of the sphere and the
dataset cause the effect of curvature to even out in the spherical case $S_1$.
The cylinder $S_0$ has local geometry equal to $\RR^2$ which causes the equality
between the methods in the $c=0$ case. The hyperboloids with $c<0$ that 
can be constructed by revolving a hyperbola around its semi-minor axis
are non-symmetric causing an increase in variance as the first 
principal direction approaches the
hyperbolic axis. The effect increases with the curvature causing 
the first principal direction to align with the hyperbolic axis for large negative 
values of $c$. That the non-linearity is quite complex can be seen from the
decreasing differences for $c=-4,-5$, a consequence of the increasing variance captured
using linearized PGA. This is caused by geodesics close to the semi-minor axis
being curved upwards towards the hyperbolic axis for large negative $c$. This
results in increased captured variance that dominates the otherwise decreasing trend as $c$ drops
below $-3$.
For all negative values of $c$, exact PGA is able to capture 
more variance in the subspace spanned by the first principal direction than
linearized PGA.

Differences between the maximal variance PGA formulation
\eqref{eq:pga-cost} and the formulation that minimizes residual errors can be
exemplified on simple geometries when the spread of the data is large.
Similar examples for Geodesic PCA with variance formulation is reported \cite{huckemann_intrinsic_2010}.
In Figure~\ref{fig:sphere-orth}, points are sampled along a great circle through
the north pole on a sphere ($c=1$). In order to illustrate the result of
maximizing projection variance, we start with the PGA center point fixed to
the north pole. In this case, each iteration of the optimization procedure
pushes the first principal component $v^1$ \emph{away} from the direction of the great
circle. In fact, the optimal direction is \emph{orthogonal} to the direction of
the great circle. This very counter-intuitive effect is caused by the projection of the
points on the southern hemisphere moving closer to the south pole as the principal
subspaces moves away from the great circle thus causing the measured variance to
increase. In fact, the cost function \eqref{eq:pga-cost}
is non-differentiable at the optimal direction and the projections become
discontinuous functions of $v^1$. If we instead choose the formulation that
minimizes residuals, the first principal component will align with the direction
of the great circle.
\begin{figure}[h]
    \begin{center}
      \subfigure[Top view of the sphere $S_{1}$.\label{fig:sphere-orth1}]%
      {\includegraphics[width=0.45\columnwidth,trim=30 10 10 00,clip=true]{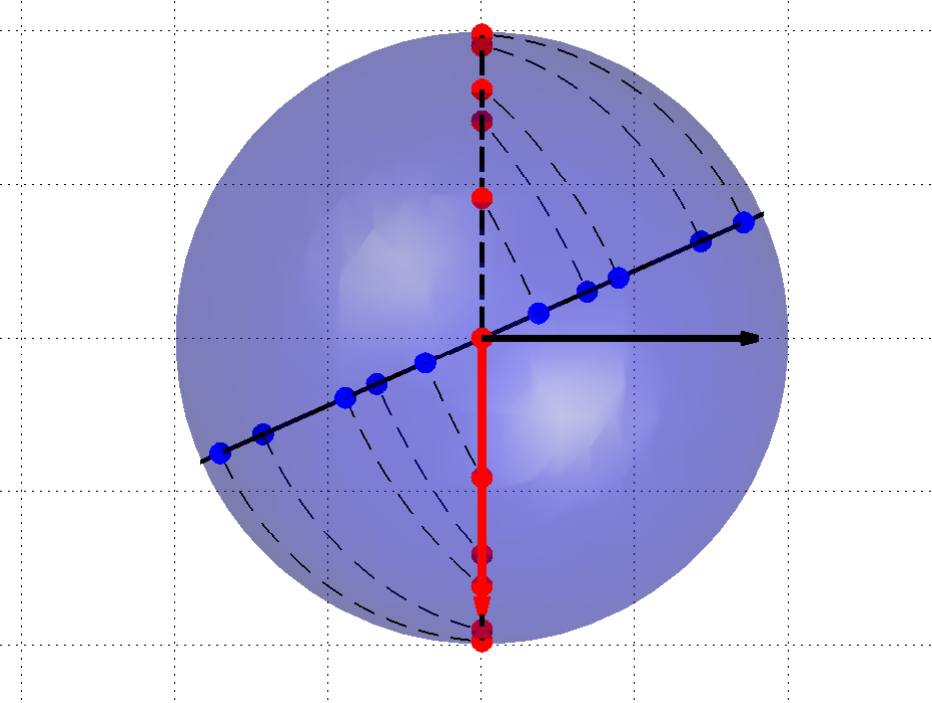}}
      \subfigure[Southern hemisphere of $S_{1}$.\label{fig:sphere-orth2}]%
      {\includegraphics[width=0.45\columnwidth,trim=30 10 10 00,clip=true]{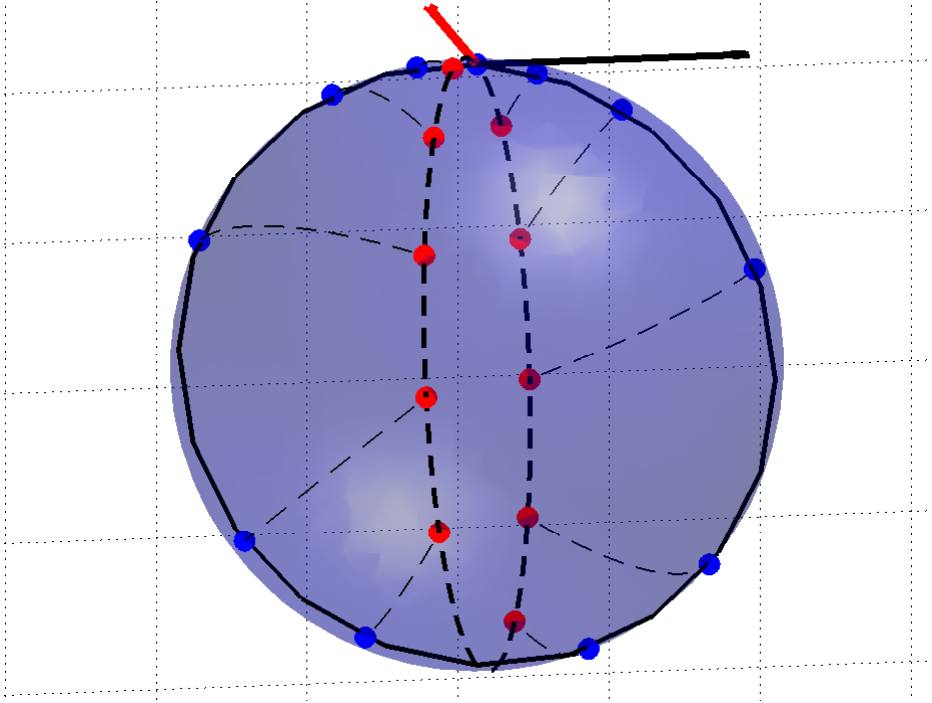}}
    \end{center}
    \caption{The sphere $S_1$ with points (red) sampled along a great circle
    (black dotted circle) with tangent vector (red arrow). The optimization for the first principal 
    component $v^1$ is stopped before it reaches its optimum (25
    iterations). The actual optimum (black arrow) is orthogonal to the
    great circle containing the data points. The variance is measured for the points (blue) projected to
    the current guess for the first principal subspace (black circle). 
    As the guess for $v^1$ moves away from being tangent to the circle containing the data points,
    points on
    the southern hemisphere move southwards 
    causing the measured variance to increase.
}
    \label{fig:sphere-orth}
\end{figure}
To show that this effect persists under permutations of the data,
we sample points uniformly along a geodesic on an ellipsoid ($c=0.5$)
adding Gaussian noise on the component orthogonal to the geodesic (std. dev.
$0.1$). This time, we optimize for the mean. The ellipsoidal geometry forces the mean
to be close to the geodesic which is the reason for sampling on an ellipsoid; on a
sphere, the mean is unstable under permutations of the data when the data lies
close to a great circle. In
Figure~\ref{fig:ellipsoid-orth}, we show the first principal component as
computed with the variance and residual formulation, respectively.
As for the example on the sphere, the optimization
converges to a first principal component orthogonal to the geodesic with the
variance formulation. We again stop the optimization after a number of iterations
before it reaches its non-differentiable optimum.
With the residual formulation, the first principal
component aligns with the geodesic along which the points are sampled.
\begin{figure}[h]
    \begin{center}
      \subfigure[First principal component (black arrow), variance formulation.\label{fig:ellipsoid-orth1}]%
      {\includegraphics[width=0.45\columnwidth,trim=100 70 100 50,clip=true]{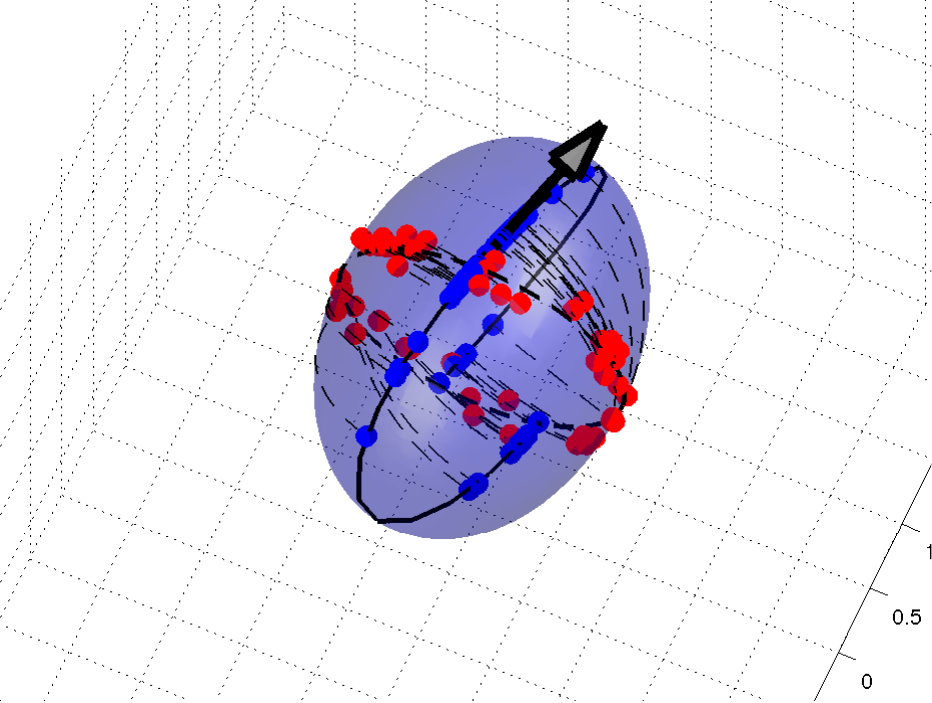}}
      \subfigure[First principal component (black arrow), residual formulation.\label{fig:ellipsoid-orth2}]%
      {\includegraphics[width=0.45\columnwidth,trim=100 70 100 50,clip=true]{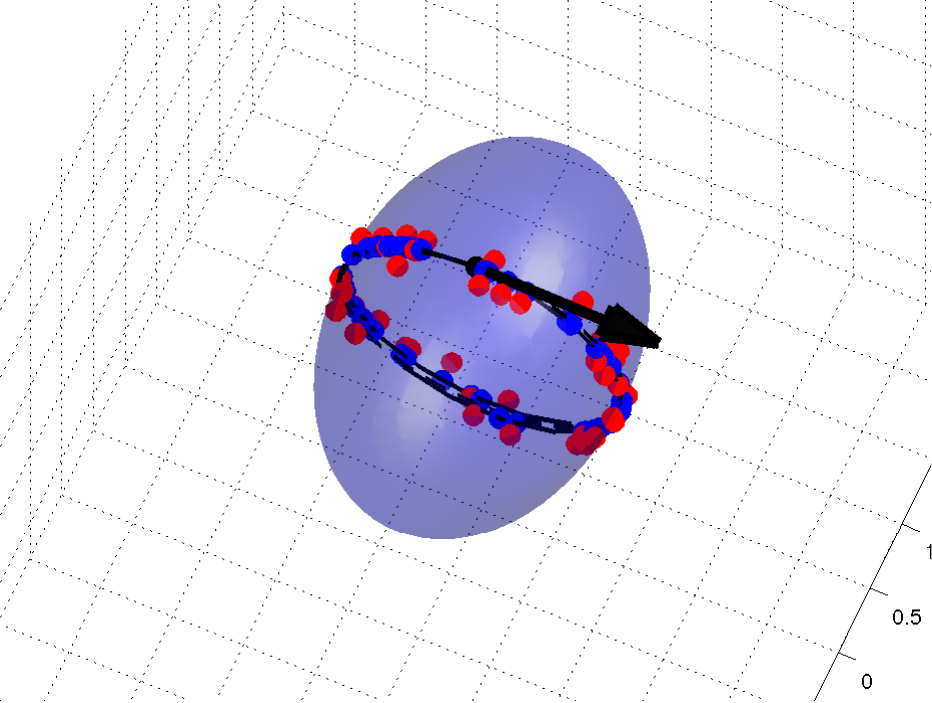}}
    \end{center}
    \caption{An ellipsoid $(c=0.5)$ with points (red) sampled uniformly along a geodesic
    (black dotted circle). Gaussian noise
    (std. dev. 0.1) is added to displace the points orthogonally to the
    geodesic.
    The black arrows show
    the result of the optimization with (a) the variance formulation and (b) the
    residual formulation. 
    With the variance formulation, the optimization is again
    stopped before it reaches its non-differential optimum orthogonal to
    the geodesic along which the points are sampled.
    The results show that the orthogonality of the first principal component
    observed in Figure~\ref{fig:sphere-orth} also occurs with perturbed data.}
    \label{fig:ellipsoid-orth}
\end{figure}
See also \cite{huckemann_intrinsic_2010} for futher
discussions on variance vs. residual formulations.

To investigate the difference between exact and linearized PGA with more than one principal direction, we consider
a four dimensional manifold embedded in $\RR^5$ and defined by
\begin{equation*}
    M_4=\{(x_1,x_2,x_3,x_4,x_5)|x_1^2-2x_2^2+x_3^2-2x_4^2+x_5^2=1\}
    \ .
\end{equation*}
We make the situation more realistic than in the previous experiment by 
sampling 32 random points in the tangent space $T_pM_4$, $p=(0,0,0,0,1)$. Since $T_pM_4$ is
an affine subspace of $\RR^5$ orthogonal to the $x_5$ axis, we can identify it with
$\RR^4$ by the map $(x_1,x_2,x_3,x_4)\mapsto(x_1,x_2,x_3,x_4,1)$. We use this
identification when sampling by defining a normal distribution in $\RR^4$,
sampling the 32 points from the distribution, and mapping the results to $T_pM_4$.
The covariance is set to
$\Sigma=\mbox{diag}(2,1,2/3,1/3)$ to get non-spherical distribution and to
increase the probability of data spreading over high-curvature parts of the
manifold. Table~\ref{table:low2} lists the variances and variance differences for 
the four principal directions for both methods along with angular differences. 
The lower variance for exact PGA compared to the linearized method for the
2nd principal direction is due to the greedy definition of PGA; when maximizing variance
for the 2nd principal direction, we keep the first principal direction
fixed. Hence we may get lower variance than what is obtainable if we were to maximize
for both principal directions together.
\begin{table}[ht]
\begin{center}
\begin{tabular}{lrrrrr}
  \hline
  \bf{Princ. comp.}: & \multicolumn{1}{c}{\bf{1}} & \multicolumn{1}{c}{\bf{2}} &
  \multicolumn{1}{c}{\bf{3}}& \multicolumn{1}{c}{\bf{4}} \\
  \hline
  angle $(^\circ)$:     & 10.1 & 10.6 & 12.0 & 12.2 \\
  linearized var.:         & 1.58 & 3.86 & 4.13 & 4.35 \\
  exact var.:           & 1.93 & 3.85 & 4.24 & 4.35 \\
  difference:           & 0.35 & -0.01 & 0.11 & 0.00 \\
  difference (\%):      & 21.9 & -0.3 & 2.6 & 0.0 \\
   \hline
\end{tabular}
\caption{Differences between the methods on $M_4$. The variances 
of the data projected to the subspaces spanned by the first $k$ principal 
directions and the percentage and angular differences are shown for $k=1,\ldots,4$.}
\label{table:low2}
\end{center}
\end{table}

We clearly see angular differences between the principal directions. In
addition, there is significant difference in accumulated variance in the 
first and third principal
direction. We note that the percentage difference is calculated from 
what corresponds to the accumulated eigenspectrum in PCA. The percentage difference of the increase between the
second and third principal direction, corresponding to the squared length of the third eigenvalue in
PCA, is greater.

\subsection{Curvature and Conjugate Points}
Again considering the surfaces $S_c$,
we can approximate the sectional curvature $K_p$ of
$S_c$ at $p$ using \eqref{eq:sec-curve}.
The approximation is dependent on the value of the positive scalar $t$ with 
increasing precision as $t$ decreases to zero.
Table~\ref{table:Ks} shows the result of the sectional curvature approximation for 
two values of $t$ compared to the real sectional curvature.
\begin{table}[ht]
    \scriptsize
    \setlength{\tabcolsep}{5pt}
\begin{center}
\begin{tabular}{lrrrrrr}
  \hline
  \bf{c}: & 
  \multicolumn{1}{c}{\bf{1}} &
  \multicolumn{1}{c}{\bf{0}} &
  \multicolumn{1}{c}{\bf{-1}} & 
  \multicolumn{1}{c}{\bf{-2}} &
  \multicolumn{1}{c}{\bf{-3}}\\
  \hline
  $K_p$:        
  & 1   & 0  & -1  & -2  & -3  \\
  $K_p$ est., $t=0.01$:  
  & 1.000   & 0.000  & -1.000  & -2.000  & -3.000  \\
  $K_p$ est., $t=0.1$:  
  & 1.000   & 0.000  & -1.001  & -2.002  & -3.005  \\
   \hline
\end{tabular}
\caption{Sectional curvature at $p$ for different values of $c$.}
\label{table:Ks}
\end{center}
\end{table}

Now let $J_t$ be the Jacobi field with $J_0=0$ and $\Df{J_0}{t}=(1,0,0)^T$ along the geodesic
$x_t=\Exp_pt(0,1,0)^T$. Figure~\ref{fig:conjugate} shows $\|J_t\|$ for different values of $c$.
We see that $\|J_\pi\|=0$ for the spherical case $S_1$ showing that $x_1$ is a conjugate point and
hence giving the upper bound $\pi$ on the injectivity radius.
The situation is illustrated in Figure~\ref{fig:sphere-jacobi}.
The local geometric equivalence between the cylinder $S_0$ and $\RR^2$ causes the
straight line for $c=0$. 
For all $c\le 1$, the injectivity radius of $S_c$ is $\pi$, but for $c<1$, the point $x_\pi$ 
is not a conjugate point\footnote{
For $c<1$, $x_\pi$ is a \emph{cut} point 
\cite[Chap.  13]{do_carmo_riemannian_1992}.
}. By looking at $\|J_t\|$, we are only able to detect
conjugate points and hence, with this experiment, we only get the bound on the injectivity radius for
$c\ge 1$. For $c>1$ the injectivity radius decreases below $1$ as seen in the
case $S_2$ with $\|J_{\tilde{t}}\|=0$ for $\tilde{t}\approx\pi/\sqrt{2}$.
\begin{figure}[t]
    \begin{center}
      \includegraphics[width=0.50\columnwidth,trim=30 20 10 10,clip=true]{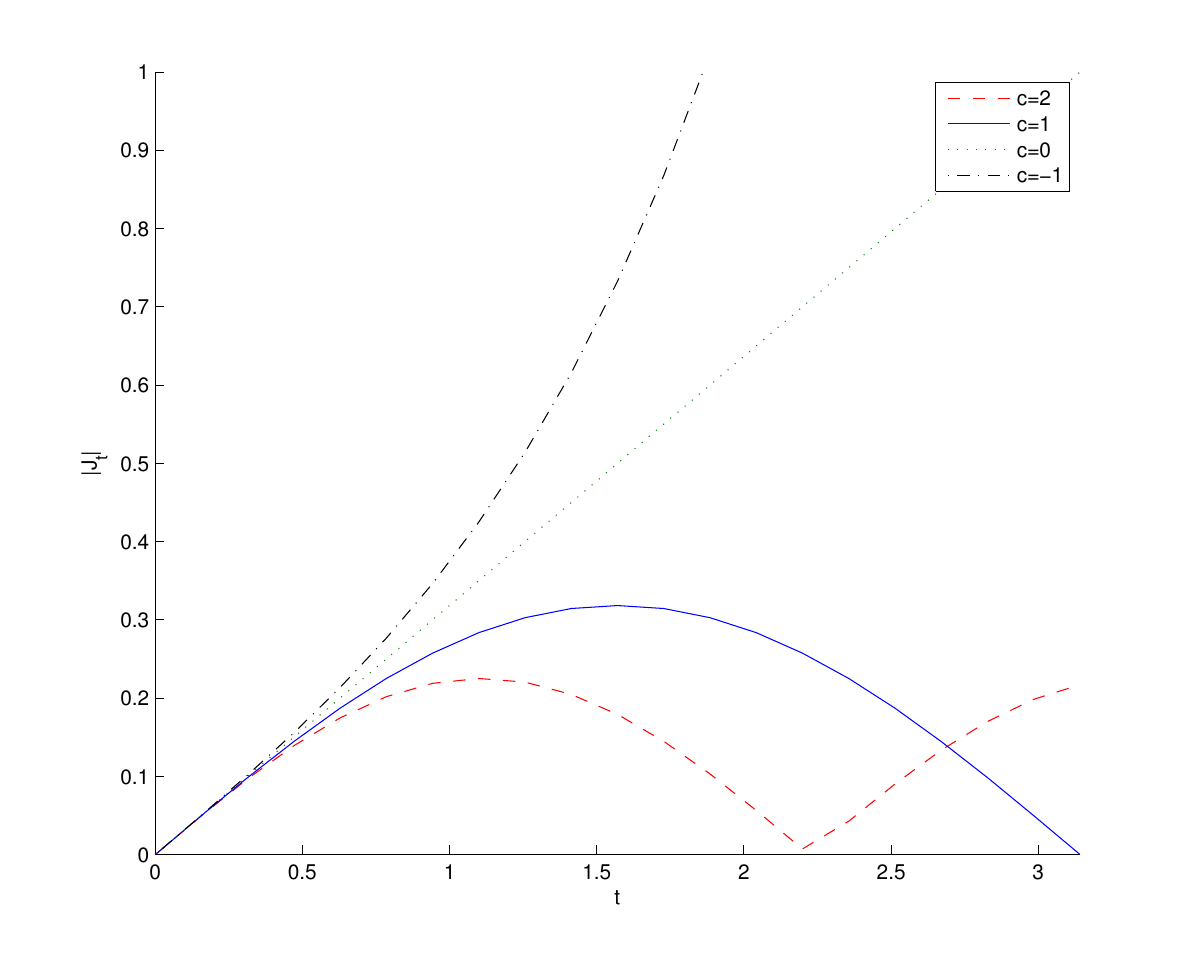}
    \end{center}
    \caption{$\|J_t\|$ for $c=2,1,0,-1$ when $J_0=0$, $\Df{J_0}{t}=(1,0,0)^T$, and $x_t=\Exp_pt(0,1,0)^T$.}
    \label{fig:conjugate}
\end{figure}

\section{Conclusion and Outlook}
Optimization problems over geodesics can be solved by constructing 
IVPs for numerical computation of Jacobi fields and second order differentials.
We use this to develop an algorithm for numerically computing exact Principal 
Geodesic Analysis and thereby eliminating the need for the traditionally
used linear approximations. In addition, the numerically computed Jacobi fields
allow injectivity radii bounds and estimation of sectional curvatures partially
solving an open problem stated in \cite{huckemann_intrinsic_2010}.

We use the developed algorithm to explore examples of manifold valued datasets where 
the principal subspaces computed by exact PGA differs from linearized PGA, and we show 
how the differences 
depend on the curvature of the manifolds and which 
formulation of PGA is used. In addition, we 
approximate sectional curvatures and bound injectivity radii and evaluate the
computed results.

We are currently extending the methods to work for quotient manifolds
$M/G$ and thereby allowing the similar computations to be performed on practically all
commonly occurring non-triangulated manifolds.
We expect this would allow Geodesic PCA to be computed on general 
quotient manifolds as well.
In addition, we are working on giving a theoretical treatment of the differences
between the variance and residual formulations of PGA. Finally, we
expect to use the automatic computation of sectional curvatures to investigate
further the effect of curvature on exact PGA and other 
statistical methods for manifold valued data.

\section*{Acknowledgements}
The authors would like to thank P. Thomas Fletcher for fruitful discussions on
the computation of exact PGA and Nicolas Courty for important remarks on
problems regarding data locality. 

\appendix

\section{Notation}
\label{app:notation}
In general, the paper follows the notation in \cite{do_carmo_riemannian_1992}.
Subscripts are used for curves on $M$ dependent on a parameter,
e.g. the curve $\alpha_t$ is a map $(-\epsilon,\epsilon)\rightarrow M$.
The subscript notation should 
not be confused with differentiation with respect to the parameter $t$.
When a local parametrization $\param{x}:U\subset\RR^\eta\rightarrow M$ is
available, it is often used to represent a
curve $\alpha_t$ so that $x_t=(x^1_t,\dots,x^\eta_t)$ is a curve in $U$ satisfying $\param{x}^{-1}\circ\alpha_t=x_t$.

The derivative $\df{\alpha_t}{t}$ of the curve $\alpha_t$
evaluated at $\tilde{t}$ belongs to the tangent space $T_{\alpha_{\tilde{t}}}M$.
The shorthand $\df{\alpha_{\tilde{t}}}{t}$ will be used for such vectors, i.e.
$\df{\alpha_t}{t}|_{t=\tilde{t}}$. In addition, when differentiating curves with
respect to $t$, we often use the shorthand $\dot{\alpha}_t$. With these conventions,
$\df{\alpha_t}{t}|_{t=0}$, the initial velocity of the curve $\alpha_t$, will be written
$\dot{\alpha}_0$.

Let $\tm f$ denote the differential of a map $f:M\to N$ and write $\tm_p f$ for
the differential evaluated at $p\in M$. When bases 
for $T_pM$ and $T_{f(p)}M$ are specified, or when $M$ and $N$ are Euclidean spaces, 
$Df$ is used instead of $df$.
For maps on product manifolds, e.g. $(v,w)\mapsto
g(v,w):M\times
\tilde{M}\to N$, we will need to distinguish differentiation with respect
to one of the variables only. Letting one of the parameters have a fixed value
$w_0$, the differential of the restricted function $v\mapsto g(v,w_0)$ from $M$
to $N$
evaluated at $v_0$ is denoted $\tm^v_{(v_0,w_0)} g$. Similarly, if $V$ is a submanifold of $M$, the differential
of $f|_V: V\to N$ will be denoted $\tmp{v\in V}f$ and its evaluation at $v_0\in V$ will
be written $\tmp{v\in V}_{v_0}f$.

When defined, the inverse of the exponential map $\Exp_q$ is the
logarithm map denoted $\Log_q(\tilde{q})$. Subsets $\Exp_qB_r(0)$ of $M$ with
$B_r(0)$ being a ball in $T_qM$ and with the radius $r>0$ sufficiently small
are examples of neighborhoods of $q$ in which 
$\Log_q(\tilde{q})$ is defined. Whenever the $\Log$-map is used, 
we will restrict to such neighborhoods without explicitly mentioning it.

When $h:M\to\RR$ is a real valued function,the
gradient of $h$ with respect to the metric is denoted $\grad\,h$, i.e. 
$\grad\,h$ satisfies $d_ph(v) = \ip{\grad_p h,v}$ for all $v\in
T_pM$. Whenever a basis of $T_pM$ is specified, or when $M$ is Euclidean, we switch 
to the usual notation $\nabla h$. Similarly, the Hessian of
$h$ is defined by the relation $\text{Hessian}(h)X = \nabla_X \grad\,h$ for all
vector fields $X$ using the covariant derivative $\nabla_X$. Again, when a basis of $T_pM$ is specified, or when $M$ is
Euclidean, the usual notation $\Hess(h)$ will be used.

\section{Expressions for the Derivative ODEs}
\label{app:A}
Because we will work with curves on manifolds that are either embedded in a
Euclidean space or where local parametrizations are available, we can perform
the derivations needed for the differential systems in Euclidean
spaces: the embedding space $\RR^m$ for the implicit case, and the parameter
space $U\subset\RR^\eta$ when a parametrization $\param{x}:U\rightarrow M$ is
available. The tensors we construct below will be tensors on the Euclidean
spaces $\RR^\eta$ or
$\RR^m$; they will be used as a compact notational representations,
and we do not attempt to give them intrinsic geometric
interpretations. The tensors will be embedding or coordinate \emph{dependent}; this is
by construction, and the tensors are thereby inherently different from
intrinsic and coordinate independent tensors such as the curvature endomorphism.

The notation will as far as possible follow the tensor notation used in
\cite{do_carmo_riemannian_1992}; however, we again note that
we use the embedding or parametrization to define the tensors 
on Euclidean domains. We will use the common identification 
between tensors and multilinear maps, i.e.
the tensor $T:(\RR^k)^r\rightarrow\RR$ defines a map multilinear map
$\tilde{T}:(\RR^k)^{r-1}\rightarrow\RR^k$ by
$\ip{\tilde{T}(y_1,\ldots,y_{r-1}),y_r}=T(y_1,\ldots,y_r)$. We will not
distinguish between a tensor and its corresponding multilinear map, and hence,
in the above case, write $T$ for both maps. 

For $s$-dependent vector fields $v_{s,1},\ldots,v_{s,r}$ and tensor field $T_s$,
we will use the equality
\begin{equation}
    \begin{split}
    &\df{T_0(v_{0,1},\ldots,v_{0,r})}{s}
    \\
    &\quad=
    \left(\df{T_0}{s}\right)(v_{0,1},\ldots,v_{0,r})
    +
    T_0(\df{v_{0,1}}{s},\ldots,v_{0,r})
    +\cdots+
    T_0(y_{0,1},\ldots,\df{v_{0,r}}{s})
    \end{split}
    \label{eq:t-deriv}
\end{equation}
for the derivative with respect to $s$.
If $T_{x_s}$ is a composition of an $z$-dependent tensor field $T_z$ and an
$s$-dependent curve $x_s$, the derivative $\df{T_{x_s}}{s}$ equals the
covariant tensor derivative $\nabla_{\sdf{x_s}{s}}T_{x_s}$ \cite[Chap.  4]{do_carmo_riemannian_1992}.
Since we will only use tensors on Euclidean spaces, such tensor derivatives 
will consist of component-wise derivatives.

In the following, when a parametrization $\param{x}$ is available, we let $T_z^P$ be the $z$-dependent $3$-tensor on
$\RR^\eta$ defined by
\begin{equation*}
    T_z^P(v_1,v_2,v_3)
    =
    -\sum_{i,j,k}^\eta\Gamma_{ij}^k(z)v_1^iv_2^jv_3^k
\end{equation*}
such that the $k$th component of $T_{x_t}^P(\dot{x}_t,\dot{x}_t)$
equals the right hand side of \eqref{eq:param-geo}. Note that $T_z^p$ is symmetric
in the first two components since the Christoffel symbols are symmetric in $i$
and $j$. Similarly, in the implicit case, we let
the $z$-dependent $3$-tensor $T_z^{I,p}$ and $2$-tensor $T_z^{I,x}$ equal the right hand side of the $p$ and
$x$ parts of \eqref{sys:impl-geo}, respectively:
\begin{align*}
    &T_z^{I,p}(v_1,v_2)
    =-\left(\sum_{k=1}^n\mu^k(z,v_1)\Hess_{z}(F^k)\right)v_2\ ,\\
    &T_z^{I,x}(v)
    =
    \left(I-D_{z}F^\dagger D_{z}F\right)v
    \ .
\end{align*}

The derivation below of \eqref{sys:JI} concerns the implicit case.
\\\ \\
    To derive $F_{q,v}^I$, we let $x_{t,s}$ be a family of geodesics with
    $x_{t,0}=x_t$, and define
    $q_s=x_{0,s}$ and $v_s=\dot{x}_{0,s}$. Assuming $\df{q_0}{s}=u$ and $\df{v_0}{s}=w$,
    the Jacobi field $J_t$ equals $\df{\Exp_{q_s}(tv_s)}{s}|_{s=0}$, 
    and, therefore, we can obtain $J_t$ by differentiating the 
    geodesic system \eqref{sys:impl-geo}. Since $M$ is embedded in $\RR^m$, we
    consider all curves and vectors to be elements of $\RR^m$.

    We use the map $\mu$ of section~\ref{sec:geodesic-systems}
    to define the tensors
    \begin{align*}
        &T_z^\mu(v)=\mu(z,v)\ ,\ 
        T_z^H(v_1,v_2)=-\left(\sum_{k=1}^nv_1^k\Hess_{z}(F^k)\right)v_2 \ , \\
        &T_z^D(v)=\left(D_zF\right)v\ ,\mbox{ and }
        T_z^{D^\dagger}(v)=\left(D_zF\right)^\dagger v .
    \end{align*}
    Note, in particular, that $T_z^{I,p}(v_1,v_2)=T_z^H(T_z^\mu(v_1),v_2)$.
    In addition, we will use the
    notation $\Lambda(A,B)$ for the right hand side of equation
    \eqref{eq:decell-expr}
    so that the derivative of a generalized inverse can be written
    $\df{(A_s^\dagger)}{s}=\Lambda(A_s,\df{A_s}{s})$.
    We claim that $\df{\Exp_{q_s}(tv_s)}{s}|_{s=0}$ equals the $z$-part of the solution of
    \eqref{sys:JI} with
    \begin{equation}
        \begin{split}
            &F_{q,v}^I\left(t,
            \begin{pmatrix}
                y_t\\
                z_t
            \end{pmatrix}
            \right)
            \\
            &\quad
            =
            {\scriptstyle
            \begin{pmatrix}
                T_{x_{t}}^{I,p}(p_t,\dot{z}_t)
                +
                \nabla_{z_t}T_{x_t}^H(T_{x_t}^\mu(p_t),\dot{x}_t)
                +
                T_{x_{t}}^H(
                T^\mu_{x_{t}}(y_t)
                -\Lambda(T_{x_{t}}^D,\nabla_{z_t}T_{x_t}^D)^Tp_t
                ,
                \dot{x}_t)
                \ \\
                T_{x_{t}}^{I,x}(y_t)
                -\Lambda(T_{x_{t}}^D,\nabla_{z_t}T_{x_t}^D)T_{x_{t}}^D(p_t)
                -T_{x_{t}}^{D^\dagger}\nabla_{z_t}T_{x_t}^D(p_t)
            \end{pmatrix}
            }
            \ .
        \end{split}
        \label{eq:FI}
    \end{equation}
    Here $p_t=p_{t,0}$ where $p_{t,s}$ are the $p$-parts of the solutions to \eqref{sys:impl-geo} 
    with initial conditions $q_s$ and $v_s$. To justify the claim, we differentiate the system \eqref{sys:impl-geo}. 
    Using \eqref{eq:t-deriv}, we get
    \begin{align*}
        &\df{\df{p_{t,0}}{s}}{t} 
        =
        \df{\dot{p}_{t,0}}{s}
        =
        \df{T_{x_{t,0}}^{I,p}}{s}(p_{t,0},\dot{x}_{t,0})
        \\
        &\qquad\quad\enspace
        = 
        \nabla_{\sdf{x_{t,0}}{s}}T_{x_t}^H(T_{x_t}^\mu(p_t),\dot{x}_t)
        +
        T_{x_{t}}^H(\nabla_{\sdf{x_{t,0}}{s}}T_{x_t}^\mu(p_{t})
        +
        T_{x_t}^\mu(\df{p_{t,0}}{s}),\dot{x}_{t})
        \\
        &\qquad\qquad\enspace\,
        +
        T_{x_{t}}^{I,p}(p_{t},\df{\dot{x}_{t,0}}{s})
    \end{align*}
    and
    \begin{align*}
        &\df{\df{x_{t,0}}{s}}{t} 
        = \df{\dot{x}_{t,0}}{s}
        = 
        \df{T_{t,0}^{I,x}}{s}(p_{t,0})
        =
        \nabla_{\sdf{x_{t,0}}{s}}T_{x_t}^{I,x}(p_{t})
        +
        T_{x_{t}}^{I,x}(\df{p_{t,0}}{s})
    \ .
    \end{align*}
    Note that the tensor derivative $\nabla_{\sdf{x_{t,0}}{s}}T_{x_t}^H$
    consists of derivatives of $H_{x_t}(F^k)$. Both the derivatives
    $\nabla_{\sdf{x_{t,0}}{s}}T_{x_t}^\mu$ and 
    $\nabla_{\sdf{x_{t,0}}{s}}T_{x_t}^{I,x}$ involve derivatives of generalized
    inverses. Therefore, we apply \eqref{eq:decell-expr} to differentiate $T_{x_t}^\mu$
    and get that
    \begin{equation*}
        \nabla_{\sdf{x_{t,0}}{s}}T_{x_t}^\mu
        =
        -\Lambda(T_{x_{t}}^D,\nabla_{\sdf{x_{t,0}}{s}}T_{x_t}^D)^T
        \ .
    \end{equation*}
    The tensor derivative $\nabla_{\sdf{x_{t,0}}{s}}T_{x_t}^D$ consists of derivatives of $D_{x_{t,s}}F$.
    Similarly,
    \begin{equation*}
        \nabla_{\sdf{x_{t,0}}{s}}T_{x_t}^{I,x}
        =
        -\Lambda(T_{x_{t}}^D,\nabla_{\sdf{x_{t,0}}{s}}T_{x_t}^D)T_{x_{t}}^D
        -T_{x_{t}}^{D^\dagger}\nabla_{\sdf{x_{t,0}}{s}}T_{x_t}^D
        \ .
    \end{equation*}
    By differentiating the initial conditions, we get
    \eqref{sys:JI} with $y=\df{p_{t,0}}{s}$, $z=\df{x_{t,0}}{s}$, and
    $F_{q,v}^I$ as defined in \eqref{eq:FI}.

As noted, we can obtain an IVP in the parametrized
case using a similar procedure. Let $\alpha_t$ be a geodesic in the $C^3$ manifold $M$. We assume 
$\param{x}:U\rightarrow M$ is a local parametrization containing $\alpha_t$, and
we let $x_t$ be the curve in $U$ representing $\alpha_t$, i.e. $\param{x}^{-1}\circ\alpha_t=x_t$. 
Let $\alpha_0=q$ and $\dot{\alpha}_0=v$, and let $u,w$ be vectors in $T_qM$. We
associate $TM$ with $\RR^\eta$ using $\param{x}$.
The Jacobi field $J_t$ along $\alpha_t$ with $J_0=u$ and
$\Df{J_0}{t}=w$ can then be found as the $z$-part of the solution of the IVP
    \begin{equation}
        \begin{split}
            &
            \begin{pmatrix}
                \dot{y}_t\\
                \dot{z}_t
            \end{pmatrix}
            =
            F_{q,v}^P\left(t,
            \begin{pmatrix}
                y_t\\
                z_t
            \end{pmatrix}
            \right)
            \ ,\\
            &
            \begin{pmatrix}
                y_0\\
                z_0\\
            \end{pmatrix}
            =
            \begin{pmatrix}
                w\\
                u
            \end{pmatrix}
            \ ,
        \end{split}
        \label{sys:JP}
    \end{equation}
    with $F_{q,v}^P$ the map constructed below.

    To derive $F_{q,v}^P$, we let $\alpha_{t,s}$ be a family of geodesics with
    $\alpha_{t,0}=\alpha_t$, and define
    $q_s=\alpha_{0,s}$ and $v_s=\dot{\alpha}_{0,s}$. Let $x_{t,s}$ represent
    $\alpha_{t,s}$ using $\param{x}$.
    Again assuming $\df{q_0}{s}=u$ and $\df{v_0}{s}=w$,
    we can obtain $J_t$ by differentiating the 
    geodesic system \eqref{eq:param-geo}. Using \eqref{eq:t-deriv} and symmetry of
    $T_z^P$, we have
    \begin{equation}
        \begin{split}
            &\df{\df{x_{t,0}}{s}}{t^2}
            =
            \df{\ddot{x}_{t,0}}{s}
            =
            \df{T_{x_{t,0}}^P(\dot{x}_{t,0},\dot{x}_{t,0}}{s})
            \\
            &\qquad\qquad
            =
            \nabla_{\sdf{x_{t,0}}{s}}T_{x_t}^P(\dot{x}_t,\dot{x}_t)
            +2T_{x_{t,0}}^P(\df{\df{x_{t,0}}{s}}{t},\dot{x}_t)\ ,\\
            &\df{x_{0,0}}{s}=u,\ \df{\df{x_{0,0}}{s}}{t}=w
        \end{split}
    \end{equation}
    because $x_{t,s}$ are solutions to \eqref{eq:param-geo} with initial conditions
    $q_s$ and $v_s$. Therefore, setting 
    $y_t=\df{\df{x_{t,0}}{s}}{t}$
    and 
    $z_t=\df{x_{t,0}}{s}$,
    we get \eqref{sys:JP} with
    \begin{align*}
        &F_{q,v}^P(t,
            \begin{pmatrix}
                y_t\\
                z_t
            \end{pmatrix}
        )
        =
        \begin{pmatrix}
            \nabla_{z_t}T_{x_t}^P(\dot{x}_t,\dot{x}_t)
            +2T_{x_t}^P(y_t,\dot{x}_t)
            \\
            y_t
        \end{pmatrix}
        \ .
    \end{align*}
    As noted above, the derivative $\nabla_{\sdf{x_{t,0}}{s}}T_{x_s}^P$ consists of 
    just the component-wise
    derivatives of $T_z^P$, i.e. the derivatives of the Christoffel symbols.

\ \\
For deriving the second order differentials, we will 
need second order derivatives of generalized inverses.
Let $A_{t,s}$ be an $s$- and $t$-dependent matrix of full rank.
From repeated application of the product rule and
\eqref{eq:decell-expr}, we see that 
when the $s$- and $t$-dependent matrices $A_{t,s}$ and $A_{t,s}^\dagger$ are 
differentiable with respect to both variables
and the mixed partial derivative $\pfcc{A_{t,s}}{s}{t}$ exists then
$\tfrac{\partial^2}{\partial s\partial t}(A_{t,s}^\dagger)
=
\tilde{\Lambda}(A_{t,s},\pfc{A_{t,s}}{t},\pfc{A_{t,s}}{s},\pfcc{A_{t,s}}{s}{t})$
where
\begin{equation}
   \begin{split}
    &\tilde{\Lambda}(A,B,C,D)
    =
    -\Lambda(A,C)BA^\dagger
    -A^\dagger DA^\dagger
    -A^\dagger B\Lambda(A,C)
    -\Big(\Lambda(A,C)A
    +A^\dagger C\Big)B^T(A^\dagger)^TA^\dagger
    \\
    &\qquad\qquad\qquad\quad
    +(I-A^\dagger A)
     \Big(D^T(A^\dagger)^TA^\dagger
     +B^T\big(\Lambda(A,C)^TA^\dagger
     +(A^\dagger)^T\Lambda(A,C)\big)\Big)
   \ .
   \end{split}
     \label{eq:second-deriv}
\end{equation}

We start the derivation with the parameterized
case. We will use the tensors introduced in the beginning of this section
and for the first order differentials.
\\\ \\
    We compute the $q$ and $r$
    parts of $G_{q,v_0,w,u}^P$ separately; denote them $G_{q,v_0,w,u}^{P,q}$ and 
    $G_{q,v_0,w,u}^{P,r}$, respectively. Let $(y_{t,s}^w,z_{t,s}^w)$ be solutions to 
    \eqref{sys:JP} with IV's
    $(w,0)^T$ and along the curves $x_{t,s}$ that represents the geodesics
    $\alpha_{t,s}$. In addition, let $y_t^w$ and $z_t^w$
    denote $y_{t,0}^w$ and $z_{t,0}^w$, respectively. Let also $(y_t^u,z_t^u)$ be
    solutions to $\eqref{sys:JP}$ with IV's $(u,0)^T$ along $x_t=x_{t,0}$.
    Differentiating system \eqref{sys:JP}, we get
        $$\df{\df{(z_{t,0}^w)}{s}}{t}
        =
        \df{(\dot{z}_{t,0}^w)}{s}
        =\df{(y_{t,0}^w)}{s}$$
    and, using symmetry of the tensors,
    \begin{equation}
    \begin{split}
        &\df{\df{(y_{t,0}^w)}{s}}{t}
        =
        \df{(\dot{y}_{t,0}^w)}{s}
        =
        \df{\nabla_{z_{t,0}^w}T_{x_{t,0}}^P(\dot{x}_{t,0},\dot{x}_{t,0})}{s}
        +2\df{T_{x_{t,0}}^P(y_{t,0}^w,\dot{x}_{t,0})}{s}
        \\
        &\qquad
        =
        \nabla_{z_t^u}\nabla_{z_{t}^w}T_{x_t}^P(\dot{x}_{t},\dot{x}_{t})
        +\nabla_{\sdf{z_{t,0}^w}{s}}T_{x_{t}}^P(\dot{x}_{t},\dot{x}_{t})
        +2\nabla_{z_{t}^w}T_{x_{t}}^P(y_t^u,\dot{x}_{t})
        \\
        &\qquad\quad\,
        +2\nabla_{z_t^u}T_{x_t}^P(y_{t}^w,\dot{x}_{t})
        +2T_{x_t}^P(\df{y_{t,0}^w}{s},\dot{x}_{t})
        +2T_{x_t}^P(y_{t}^w,y_t^u)
        \ .
    \end{split}
    \label{eq:GPq}
    \end{equation}
    Therefore, letting $q_t=\df{y_{t,0}^w}{s}$ and $r_t=\df{z_{t,0}^w}{s}$, we get
    $G_{q,v_0,w,u}^{P,q}(t,(r_t\ q_t)^T)$ 
    as the right hand side of \eqref{eq:GPq}
    and $G_{q,v_0,w,u}^{P,r}(t,(r_t\ q_t)^T)$ equal to $q_t$.
    The initial values are both $0$ since $y_{0,s}^w$ and
    $z_{0,s}^w$ equal $0$ and $w$, respectively, and, therefore, are not $s$-dependent.
    
    For the implicit case, we will again compute the $r$ and $q$
    parts of $G_{q,v_0,w,u}^I$ separately. Let now $(y_{t,s}^w,z_{t,s}^w)$ 
    be solutions to \eqref{sys:JI} along the geodesics $x_{t,s}$ and with IV's
    $(w,0)^T$, and let $(y_t^u,z_t^u)$ be
    solutions to $\eqref{sys:JI}$ along $x_t$ and with IV's $(u,0)^T$.
    Let also $p_{t,s}$ denote the $p$-parts of the solutions to \eqref{sys:impl-geo} with
    initial conditions $q$ and $v_s$, and write $p_t=p_{t,0}$, $y_t^w=y_{t,0}^w$,
    and $z_t^w=z_{t,0}^w$. Recall that all curves and vectors are considered
    elements of the embedding space $\RR^m$.

    Differentiating system \eqref{sys:JI}, we get
    \begin{align*}
        &\df{\df{y_{t,0}^w}{s}}{t}
        =
        \df{\dot{y}_{t,0}^w}{s}
        =
        \df{T_{x_{t,0}}^{I,p}(p_{t,0},\dot{z}^w_{t,0})}{s}
                +
                \df{\nabla_{z^w_{t,0}}T_{x_{t,0}}^H(T_{x_{t,0}}^\mu(p_{t,0}),\dot{x}_{t,0})}{s}
        \\
        &\qquad\qquad\qquad\qquad\enspace
                +
                \df{T_{x_{t,0}}^H(
                T^\mu_{x_{t,0}}(y^w_{t,0})
                -\Lambda(T_{x_{t,0}}^D,\nabla_{z^w_{t,0}}T_{x_{t,0}}^D)^Tp_{t,0}
                ,
                \dot{x}_{t,0})}{s}
        \ .
    \end{align*}
    Using the map $\tilde{\Lambda}$ defined in \eqref{eq:second-deriv}, we have
    \begin{equation*}
        \df{\Lambda(T_{x_{t,0}}^D,\nabla_{z^w_{t,0}}T_{x_{t,0}}^D)}{s}^T
        =
        \tilde{\Lambda}(T_{x_{t}}^D,
        \nabla_{z^w_{t}}T_{x_{t}}^D,
        \nabla_{z^u_{t}}T_{x_{t}}^D,
        \nabla_{z^u_t}\nabla_{z^w_{t}}T_{x_{t}}^D
        )^T \ .
    \end{equation*}
    Combining the equations, we get
    \begin{align*}
        \df{\df{y_{t,0}^w}{s}}{t}
        &=
        \nabla_{z^u_t}T_{x_{t}}^{I,p}(p_{t},\dot{z}^w_{t})
        +T_{x_{t}}^{I,p}(y^u_t,\dot{z}^w_{t})
        +T_{x_{t}}^{I,p}(p_{t},\df{\df{z^w_{t,0}}{s}}{t})
        \\
        &\quad\,
        +\nabla_{z^u_t}\nabla_{z^w_{t}}T_{x_{t}}^H(T_{x_{t}}^\mu(p_{t}),\dot{x}_{t})
        +\nabla_{\df{z^w_{t,0}}{s}}T_{x_{t}}^H(T_{x_{t}}^\mu(p_{t}),\dot{x}_{t})
        \\
        &\quad\,
        +\nabla_{z^w_{t}}T_{x_{t}}^H(
                T^\mu_{x_{t}}(y^u_t)
                -\Lambda(T_{x_{t}}^D,\nabla_{z^u_t}T_{x_t}^D)^Tp_t
        ,\dot{x}_{t})
        +\nabla_{z^w_{t}}T_{x_{t}}^H(T_{x_{t}}^\mu(p_{t}),\dot{z}^u_{t})
        \\
        &\quad\,
        +\nabla_{z^u_t}T_{x_{t}}^H(
        T^\mu_{x_{t}}(y^w_{t})
        -\Lambda(T_{x_{t}}^D,\nabla_{z^w_{t}}T_{x_{t}}^D)^Tp_{t}
        ,
        \dot{x}_{t})
        \\
        &\quad\,
        +T_{x_{t}}^H(
        T^\mu_{x_{t}}(\df{y^w_{t,0}}{s})
        -\Lambda(T_{x_{t}}^D,\nabla_{z^u_{t}}T_{x_{t}}^D)^Ty^w_{t}
        ,
        \dot{x}_{t})
        \\
        &\quad\,
        -T_{x_{t}}^H(
        \tilde{\Lambda}(T_{x_{t}}^D,
        \nabla_{z^w_{t}}T_{x_{t}}^D,
        \nabla_{z^u_{t}}T_{x_{t}}^D,
        \nabla_{z^u_t}\nabla_{z^w_{t}}T_{x_{t}}^D
        )^Tp_t
        +
        \Lambda(T_{x_{t}}^D,\nabla_{z^u_{t}}T_{x_{t}}^D)^Ty^u_{t}
        ,
        \dot{x}_{t})
        \\
        &\quad\,
        +T_{x_{t}}^H(
        T^\mu_{x_{t}}(y^w_{t})
        -\Lambda(T_{x_{t}}^D,\nabla_{z^w_{t}}T_{x_{t}}^D)^Tp_{t}
        ,
        \dot{z}^u_{t})
        \ .
    \end{align*}
    Substituting $\df{z_{t,0}^w}{s}$ with $r_t$ and $\df{y_{t,0}^w}{s}$ with $q_t$, we get
    $G_{q,v_0,w,u}^{I,q}$ as the right hand side of the equation.
    Likewise,
    \begin{align*}
        \df{\df{z_{t,0}^w}{s}}{t}
        &=
        \df{T_{x_{t,0}}^{I,x}(y^w_{t,0})}{s}
        -\df{\Lambda(T_{x_{t,0}}^D,\nabla_{z^w_{t,0}}T_{x_{t,0}}^D)T_{x_{t,0}}^D(p_{t,0})}{s}
        -\df{T_{x_{t,0}}^{D^\dagger}\nabla_{z^w_{t,0}}T_{x_{t,0}}^D(p_{t,0})}{s}
        \\
        &=
        \nabla_{z^u_t}T_{x_{t}}^{I,x}(y^w_t)
        +T_{x_{t}}^{I,x}(\df{y^w_{t,0}}{s})
        \\
        &
        -\tilde{\Lambda}(T_{x_{t}}^D,\nabla_{z^w_{t}}T_{x_{t}}^D,
        \nabla_{z^u_{t}}T_{x_{t}}^D,
        \nabla_{z^u_{t}}\nabla_{z^w_{t}}T_{x_{t}}^D)T_{x_{t}}^D(p_{t})
        \\
        &
        -\Lambda(T_{x_{t}}^D,\nabla_{z^w_{t}}T_{x_{t}}^D)\nabla_{z^u_t}T_{x_{t}}^D(p_{t})
        -\Lambda(T_{x_{t}}^D,\nabla_{z^w_{t}}T_{x_{t}}^D)T_{x_{t}}^D(y^u_{t})
        \\
        &
        -\Lambda(T_{x_{t}}^D,\nabla_{z^u_t}T_{x_{t}}^D)\nabla_{z^w_{t}}T_{x_{t}}^D(p_{t})
        -T_{x_{t}}^{D^\dagger}\nabla_{z^u_t}\nabla_{z^w_{t}}T_{x_{t}}^D(p_{t})
        -T_{x_{t}}^{D^\dagger}\nabla_{z^w_{t}}T_{x_{t}}^D(y^u_{t})
        \ .
    \end{align*}
    Again, after substituting $\df{y_{t,0}^w}{s}$ with $q_t$ as above, we get $G_{q,v_0,w,u}^{I,r}$ as the
    right hand side of the equation. As for the parametric case, both initial
    values are zero.

\section{The Projection Differential}
\label{app:B}
For the proof of Proposition~\ref{prop:grad-proj}, we will need the following result 
to show that equation \eqref{eq:grad-proj} is independent of the chosen basis.
\begin{lemma}
    Let $S$ be an open subset of $\RR^k$ and $U:S\rightarrow M^{k\times(k-1)}$ a $C^1$ map with the property that for any
    $v\in S$, the columns of the matrix $(\tfrac{v}{\|v\|}\ U(v))$ constitute an orthonormal basis
    for $\RR^k$. Let $u_v^j$ denote the $j$th column of $U(v)$. Then for any
    $v_0\in S$ and $w\in\RR^k$,
    $\ip{\df{u_{v_0+tw}^j|_{t=0}}{t},v_0}=-\ip{u_{v_0}^j,w}$. As consequence of
    this, if $\tilde{U}:S\rightarrow\RR^{k-1}$ denotes the map $v\mapsto U(v)^T\tfrac{v_0}{\|v_0\|}$
    then
    \begin{equation*}
        D^{v\in\Span(u_{v_0}^1,\ldots,u_{v_0}^{k-1})}_{v_0}\tilde{U}(v)
        =
        -I_{k-1}
    \end{equation*}
    in the basis $u_{v_0}^1,\ldots,u_{v_0}^{k-1}$ for
    $\Span(u_{v_0}^1,\ldots,u_{v_0}^{k-1})$.
    
    \label{lem:orth-deriv}
\end{lemma}

\bibliographystyle{myamsplain}
\bibliography{bibliography.bib}

\end{document}